\documentclass[10pt,twocolumn,twoside]{IEEEtran}

\IEEEoverridecommandlockouts                              % This command is only needed if 
\pdfoutput=1 
\usepackage{times} % assumes new font selection scheme installed
\usepackage[latin1]{inputenc}
\usepackage[mathcal]{euscript}
\usepackage{amssymb,amsmath,color,amsthm}
\usepackage{graphicx}
\usepackage[english]{babel}
\usepackage[noadjust]{cite}
\usepackage{subcaption}

\makeatletter
\renewcommand\p@subfigure{\thefigure\,}

\DeclareCaptionLabelFormat{mysublabelfmt}{(\alph{sub\@captype})}
\makeatother

\newtheorem{prop}{Proposition}

\newtheorem{theo}{Theorem}
\newtheorem{lemma}{Lemma}

\newtheorem{coro}{Corollary}
\newtheorem{assum}{Assumption}

\newcounter{hypA}

\newcounter{hyp}

\newcommand{\un}{{\bs 1}}
\newcommand{\indic}{\mathbb{I}}

\newcommand{\eqdef}{:=}

\newcommand{\EE}[1]{{\mathbb{E}}\left[ #1 \right]}
\newcommand{\PP}[1]{{\mathbb{P}}\left\{ #1 \right\}}

\newcommand{\bs}{\boldsymbol}

\newcommand{\cT}{{\mathcal T}}

\newcommand{\cA}{{\mathcal A}}
\newcommand{\cK}{{\mathcal K}}
\newcommand{\cL}{{\mathcal L}}

\newcommand{\cF}{{\mathcal F}}
\newcommand{\cN}{{\mathcal N}}

\newcommand{\cW}{{\cal W}}

\newcommand{\cJ}{{\mathcal J}}
\newcommand{\cS}{{\mathcal S}}

\newcommand{\bP}{{\mathbb P}}
\newcommand{\bR}{{\mathbb R}}

\newcommand{\bE}{{\mathbb E}}
\newcommand{\E}{{\mathbb E}}

\newcommand{\bN}{{\mathbb N}}

\newcommand{\T}{^\text{\scshape{t}}}

\newcommand{\la}{\langle}
\newcommand{\ra}{\rangle}

\newcommand{\thetash}{\vartheta}

\newcommand{\bth}{{ \theta}}

\newcommand{\thn}{{\theta}_{n}}

\newcommand{\thnmu}{{ \theta}_{n-1}}

\newcommand{\athn}{\langle\thn\rangle}

\newcommand{\athnmu}{\langle\thnmu\rangle}

\newcommand{\Jo}{{J_\bot}}
\newcommand{\cJo}{{{\mathcal J}_\bot}}

\newcommand{\vect}{\text{vec}\,}

\newcommand{\F}{\mathcal{F}}
\newcommand{\Rset}{\mathbb{R}}

\newcommand{\cR}{{\mathcal R}}
\newcommand{\A}{{\mathsf A}}
\newcommand{\B}{{\mathsf B}}
\newcommand{\C}{{\mathsf C}}
\newcommand{\D}{{\mathsf D}}
\usepackage[textwidth=2.3cm, textsize=footnotesize]{todonotes}    %  [disable] shuts it off
\title{\LARGE \bf
Success and Failure of Adaptation-Diffusion Algorithms \\ for Consensus in Multi-Agent Networks}

\author{Gemma Morral$^{*}$, Pascal Bianchi and Gersende Fort% <-this % stops a space
\thanks{*This work is supported by DGA (French Armement Procurement Agency), the Institut Mines-Télécom and by the ANR grant ODISSEE of program ASTRID (ANR-13-ASTR-0030).}% <-this % stops a space
\thanks{G. Morral, P. Bianchi and G. Fort are with LTCI, Télécom ParisTech \&
  CNRS, 46 rue Barrault, 75634 Paris Cedex 13, France {\tt\small
    [firstname].[lastname]@telecom-paristech.fr}}}

\begin{document}

\maketitle
\thispagestyle{empty}
\pagestyle{empty}

%%%%%%%%%%%%%%%%%%%%%%%%%%%%%%%%%%%%%%%%%%%%%%%%%%%%%%%%%%%%%%%%%%%%%%%%%%%%%%%%
\begin{abstract}
  This paper investigates the problem of distributed stochastic approximation
  in multi-agent systems. The algorithm under study consists of two steps: a
  local stochastic approximation step and a diffusion step which drives the
  network to a consensus. The diffusion step uses row-stochastic matrices to
  weight the network exchanges.  
 As opposed to previous works, exchange matrices are not supposed to be doubly stochastic,
  and may also depend on the past estimate.

We prove that non-doubly stochastic matrices generally influence the limit points of the algorithm.
Nevertheless, the limit points are not affected by the choice of the matrices provided that the latter are doubly-stochastic 
in expectation. This conclusion  legitimates the use of broadcast-like diffusion protocols, which are 
easier to implement. Next, by means of a central limit theorem, we prove that doubly stochastic protocols perform asymptotically as well
as centralized algorithms and we quantify the degradation caused by the use of non doubly stochastic matrices.
Throughout the paper, a special emphasis is put on the special case of distributed non-convex optimization
as an illustration of our results.

\end{abstract}

%%%%%%%%%%%%%%%%%%%%%%%%%%%%%%%%%%%%%%%%%%%%%%%%%%%%%%%%%%%%%%%%%%%%%%%%%%%%%%%%
\section{INTRODUCTION}

Distributed stochastic approximation has been recently proposed using different
cooperative approaches. In the so-called \emph{incremental} approach (see for
instance \cite{rabbat:nowak:sac-2005,lop-sayed:2007,johansson:2008,ram:nedic:veeravalli:jopt-2009}) a message
containing an estimate of the quantity of interest iteratively travels all over
the network. This paper focuses on another cooperative approach based on
\emph{average consensus} techniques where the estimates computed locally by
each agent are combined through the network.

Consider a network composed by $N$ agents, or nodes. Agents seek to
find a consensus on some global parameter by means of local
observations and peer-to-peer communications. The aim of this paper is
to analyze the behavior of the following distributed algorithm. Node $i$ ($i=1,\dots,N$) generates a $\bR^d$-valued stochastic process
$(\theta_{n,i})_{n\geq 0}$. At time $n$, the update is in two steps:\\
\noindent {\tt [Local step]} Node $i$ generates  a temporary iterate
$\tilde \theta_{n,i}$ given by 
\begin{equation}
  \label{eq:tempupdate}
\tilde \theta_{n,i}= \theta_{n-1,i} + \gamma_{n}\, Y_{n,i}\ ,
\end{equation}
where $\gamma_n$ is a deterministic positive step size and where the 
$\bR^d$-valued random process $(Y_{n,i})_{n \geq 1}$ represents the 
observations made by agent $i$. \\
\noindent {\tt [Gossip step]} Node $i$ is able to observe the values 
$\tilde\theta_{n,j}$ of some other $j$'s and computes the weighted average:
\begin{equation}
  \label{eq:gossipIntro}
\theta_{n,i}=\sum_{j=1}^N w_n(i,j)\,\tilde \theta_{n,j} \ ,
\end{equation}
where the $w_n(i,j)$'s are scalar non-negative random coefficients such that
$\sum_{j=1}^N w_n(i,j)=1$ for any $i$. The sequence of random matrices $W_n
\eqdef [ w_n(i,j) ]_{i,j=1}^N$ represents the time-varying communication
network between the nodes. One simply set $w_n(i,j)=0$ whenever nodes $i$ and $j$ are unable to communicate at time $n$.
The aim of this paper is to investigate the almost sure (a.s.) convergence of this algorithm as $n$ tends to infinity
as well as the convergence rate. Our goal is in particular to quantify the effect of the sequence of matrices $(W_n)_{n\geq 1}$ on the convergence.
The algorithm is initialized at some arbitrary $\bR^{d}$-valued vectors $\theta_{0,1},\cdots,\theta_{0,N}$.

{\bf Application to distributed optimization}. The
algorithm~(\ref{eq:tempupdate})-(\ref{eq:gossipIntro}) under study is not new.
The idea beyond the algorithm traces back to
\cite{tsitsiklis:phd-1984,tsitsiklis:bertsekas:athans:tac-1986} where a network
of processors seeks to optimize some objective function \emph{known} by all
agents (possibly up to some additive noise). More recently, numerous works
extended this kind of algorithm to more involved multi-agent scenarios,
see~\cite{kushner-siam87,lop-sayed-asap06,nedic:ozdaglar:parrilo:tac-2010,kar:2010,johansson:2010,stankovic:stankovic:stipanovic:tac-2011,chen-sayed-sp12,bianchi:fort:hachem:2012,nous:eusipco,nousTAC,nedic2013distributed,bianchi:iutzeler:hachem:2013,bianchi:hachem:iutzeler:2014}
as a non-exhaustive list. In this context, one seeks to minimize a sum of local
private cost functions $f_i$ of the agents:
\begin{equation}
\label{eq:distOpt}
\min_{\theta\in {\mathbb R}^d} \  \sum_{i=1}^N f_i(\theta)\ ,
\end{equation}
where for all $i$, the function $f_i$ is supposed to be unknown by any other
agent $j, j\neq i$. To address this question, it is assumed that 
\begin{equation}
Y_{n,i} = -\nabla f_i(\theta_{n-1,i})+\xi_{n,i}
\label{eq:YnDistOpt}
\end{equation}  where $\nabla$ is the gradient
operator and $\xi_{n,i}$ represents some random perturbation which possibly
occurs when observing the gradient. 
In this paper, we handle the case where
functions $f_i$ are not necessarily convex. Of course, in that case, there is
generally no hope to ensure the convergence to a minimizer
to~(\ref{eq:distOpt}). Instead, a more realistic objective is to achieve
\emph{critical points} of the objective function \emph{i.e.}, points $\theta$
such that $\sum_i\nabla f_i(\theta)=0$.

%In a machine learning context, $f_i$ is typically the risk of
%a classifier indexed by $\theta$ (for more details we refer to \cite{forero-cano-giannakis-jmlr11,boyd2011distributed,agarwal2011reliable}).
% a classifier indexed by $\theta$ (for more details we refer to \cite{forero-cano-giannakis-jmlr11} {\color{red} \cite{agarwal-duchi-nips10} and most recently \cite{hall:McDonald:2010}. In \cite{Patra:2011} it is studied the problem of finding the optimal vector quantizer which minimizes a non convex cost function called distortion. They proposed a distributed version of the CLVQ (Competitive Learning Vector Quantization) algorithm which is an on-line implementation of the k-means. They used a stochastic gradient descent procedure with decreasing step size locally at each agent along with asynchronous combining weights according to a possible local clock of each agent as the global clock it is assumed to be unknown. Some of the assumptions are: regularities on the gradient function, row stochastic communication matrix and a distribution with compact support and bounded density, as an example. They proved the consistency of the algorithm :  the trajectories of agents reach an asymptotical consensus a.s. and the corresponding agreement vector converges a.s. towards one of the random connected component of the set of critical points. **Citer aussi Duchi, Patra, Hall**}).

Convergence to a global minimizer is shown in \cite{nedic:ozdaglar:tac-2009}
assuming \emph{convex} utility functions and bounded (sub)gradients. 
The results of \cite{nedic:ozdaglar:tac-2009} are extended in \cite{ram:nedic:veeravalli:jota-2010} to the stochastic descent case
\emph{i.e.}, when the observation of utility functions is perturbed by a random
noise. More recently,~\cite{bianchi:fort:hachem:2012} investigated distributed stochastic approximation at large,
providing stability conditions of the algorithm~(\ref{eq:tempupdate})-(\ref{eq:gossipIntro}) while
relaxing the bounded gradient assumption and including the case of random communication links.
In~\cite{bianchi:fort:hachem:2012}, it is also proved  under some hypotheses that the estimation error is asymptotically normal: the convergence rate and the asymptotic covariance matrix are characterized. An enhanced averaging algorithm \`a la Polyak is also proposed to
recover the optimal convergence rate.

{\bf Doubly and non-doubly stochastic matrices}. In most works (see for instance
\cite{nedic:ozdaglar:tac-2009,ram:nedic:veeravalli:jota-2010,rabbat:tsianos:2013}), the matrices
$(W_n)_{n \geq 1}$ are assumed \emph{doubly stochastic}, meaning that
$W_n^T\un=W_n\un=\un$ where $\un$ is the $N\times 1$ vector whose components
are all equal to one and where $^T$ denotes transposition. Although
row-stochasticity ($W_n\un=\un$) is rather easy to ensure in practice,
column-stochasticity ($W_n^T\un=\un$) implies more stringent restrictions on
the communication protocol. For instance, in~\cite{boyd:2006}, each one-way
transmission from an agent $i$ to another agent $j$ requires at the same time a
feedback link from $j$ to $i$. 
As a matter of fact, double stochasticity prevents from using natural
broadcast schemes, in which a given node may transmit its local estimate to
\emph{all} neighbors without expecting any immediate feedback.

Remarkably, although generally assumed, double stochasticity of the matrices
$W_n$ is in fact {\bf not} mandatory. A couple of works (see \emph{e.g.}
\cite{nedic:tac-2011,bianchi:fort:hachem:2012}) get rid of the
column-stochasticity condition, but at the price of assumptions that may not
always be satisfied in practice. Other works
(\cite{rabbat:tsianos:2012,nedic2013distributed,nedic:olshevsky:2014}) manage to circumvent the use
of feedback links by coupling the gradient descent with the so-called push-sum
protocol~\cite{kempe:dobra:gehrke:focs-2003}. The latter however introduces an
additional communication of weights in the network in order to keep track of
some summary of the past transmissions.  In this paper, we address the
following questions:
What conditions on the sequence $(W_n)_{n \geq 1}$ are needed to
  ensure that Algorithm~(\ref{eq:tempupdate})-(\ref{eq:gossipIntro})
  drives all agents to a common critical point of $\sum_if_i$? What  happens if these conditions are not satisfied?
How is the convergence rate influenced by the communication protocol?
\smallskip

{\bf Contributions}. 
\begin{enumerate}
\item Assuming that $(W_n)_{n \geq 1}$ forms an independent and identically distributed (i.i.d.) sequence of stochastic
  matrices, we prove under some technical hypotheses that
  Algorithm~(\ref{eq:tempupdate})-(\ref{eq:gossipIntro}) leads the
  agents to a consensus, which is characterized.  It is shown that the latter consensus
  does not necessarily coincide with a critical point of $\sum_i f_i$.

\item We provide sufficient conditions either on the communication
  protocol  $(W_n)_{n \geq 1}$ \emph{or} on the functions $f_i$ which ensure that limit
  points are the critical points of $\sum_i f_i$. 

\item When such
  conditions are not satisfied, we also propose a simple modification
  of the algorithm which allows to recover the sought behavior.
\item We extend our results to a broader setting, assuming that the matrices $(W_n)_{n \geq 1}$ are no longer i.i.d.,
but are likely to depend on both the current observations and the past estimates. 
We also investigate a general stochastic approximation framework which goes beyond the 
model~(\ref{eq:YnDistOpt}) and beyond the only problem of distributed optimization.
\item We characterize the convergence rate of the algorithm under the form of a
  central limit theorem. Unlike~\cite{bianchi:fort:hachem:2012}, we address
  the case where the sequence $(W_n)_{n \geq 1}$ is not necessarily doubly
  stochastic.  We show that non-doubly stochastic matrix have an influence on
  the asymptotic error covariance (even if they are doubly stochastic in
  average). On the other hand, we prove that when the matrix $W_n$ is doubly
  stochastic for all $n$, the asymptotic covariance is identical to the one
  obtained in a centralized setting.
\end{enumerate}

The paper is organized as follows. Section~\ref{sec:dist_opt} is a gentle presentation of our results
in the special case of distributed optimization (see~(\ref{eq:distOpt})) assuming in addition that sequence $(W_n)$ is i.i.d.
In Section~\ref{sec:general} we provide the general setting to study almost sure convergence.
Almost sure convergence is studied in Section~\ref{sec:CVanalysis}.
Section~\ref{sec:CvgRate} investigates convergence rates. Conclusions and numerical results complete the paper.

\smallskip

\emph{Notations:} Throughout the paper, the vectors are column vectors. 
The random variables $W_n \in \bR^{N \times N}$ and $Y_n \eqdef
(Y_{n,1}^T,\dots, Y_{n,N}^T)^T \in \bR^{dN}$, $n \geq 1$, are defined on the
same measurable space equipped with a probability~$\bP$; $\bE$ denotes the
associated expectation. For any $n\geq 1$, define the $\sigma$-field $\cF_n
\eqdef \sigma({ \theta}_0,W_{1},\dots, W_n,Y_{1},\dots,Y_n)$ where $
\theta_0$ is the (possibly random) initial point of the algorithm. \\
It is assumed that for any $i \in 1, \dots, N$, $(\theta_{n,i})_{n \geq 0}$ satisfies
the update equations~(\ref{eq:tempupdate})-(\ref{eq:gossipIntro}); and we set
$$
\theta_n:=(\theta_{n,1}^T,\dots,\theta_{n,N}^T)^T.
$$
For any vector $x \in
\bR^\ell$, $|x|$ represents the Euclidean norm of $x$. $I_N$ is the $N\times
N$ identity matrix.  $J \eqdef \un \un^T/N$ denotes the orthogonal projector
onto the linear span of the all-one $N\times 1$ vector $\un$, and $\Jo \eqdef
I_N-J$. 
We denote by $\otimes$ the Kronecker product between matrices. 
For a matrix $A$, the spectral norm is denoted by $\|A\|$ and the spectral radius is denoted by
$r(A)$ whenever $A$ is a square matrix.

\section{Distributed Optimization}
\label{sec:dist_opt}

\subsection{Framework}

We first sketch our result in the special case of distributed optimization
\emph{i.e.}, when the ``innovation'' $Y_{n,i}$ of the algorithm in
(\ref{eq:tempupdate}) has the form~(\ref{eq:YnDistOpt}). 

\begin{assum}
\label{hyp:distOpt}
  \begin{enumerate}
  \item \label{hyp:distOpt:item1} $f_i: \bR^d\to\bR$ is differentiable and
    $\nabla f_i$ is locally Lipschitz-continuous.
  \item \label{hyp:distOpt:item2} For any Borel set A of $\bR^{dN}$,
    $
    \bP\left[ \xi_{n+1} \in A \, |\, \cF_n\right] =\nu_{\theta_n}(A) 
    $ almost surely (a.s.)
    where $(\nu_{\theta})_{\theta \in \bR^{dN}}$ is a family of probability
    measures such that $\int z \,d\nu_{\theta}(z) = 0 $
and $\sup\limits_{\theta\in \cK}\int
      |z|^2d\nu_{\theta}(z)<\infty$ for any compact set $\cK\subset
      \bR^{dN}$\,.
  \end{enumerate}
\end{assum}
For simplicity, the
matrix-valued process $W_n$ will be assumed i.i.d. and independent of both
processes $Y_{n}$ and $\theta_n$. This assumption will be relaxed in
section~\ref{sec:general}.
%We consider the following assumption on the communication matrix $W_n$:
\begin{assum}
\label{hyp:Wiid}
\begin{enumerate}
\item \label{hyp:Wiid:item1} For any $n\geq 0$, conditionally to~$\cF_n$,
  $( Y_{n+1}, W_{n+1})$ are independent.
 
\item \label{hyp:Wiid:item2} $(W_n)_{n\geq 1}$ is an i.i.d. sequence of row-stochastic matrices (i.e.,
  $W_n\un=\un$ for any $n$) with non-negative entries.
\item \label{hyp:Wiid:item3}  The spectral radius of the matrix $\bE[W_1^TJ_\bot W_1]$ is strictly lower than $1$.
  \end{enumerate}
\end{assum}
%Assumption~\ref{hyp:Wiid} implies that at each time $n$, the communication
%matrix $W_n$ is random and does not depend on the observations $(Y_{1}, \dots, Y_n)$. 
The row-stochasticity assumption is a rather mild condition. 
% It claims that $
% \sum_j w_n(i,j)=1$ for any $i$ \emph{i.e.}, each node $i$ computes a weighted average
% of the temporary updates at each node (with possibly some null weights).  
In many works, it is also assumed that $W_n$ is column-stochastic
\emph{i.e.}, $\sum_i w_n(i,j)=1$ for any $j$, though this assumption is not required in this work.
% . Our weaker framework addresses
% more general gossip protocols, usually less demanding in terms of scheduling
% and overall network coordination.
Assumption~\ref{hyp:Wiid}-\ref{hyp:Wiid:item3}) is a contraction condition
which is required to drive the network to a consensus.

%Assumptions are also required on the step-size $(\gamma_n)_{n \geq 1}$.
\begin{assum}
  \label{hyp:step}
  The deterministic step-size sequence $(\gamma_n)_{n \geq 1}$ satisfies $\gamma_n>0$ and:
  \begin{enumerate}
  \item \label{hyp:step:item1} $\lim_n \gamma_{n+1}/\gamma_n\ = 1$,
  \item  \label{hyp:step:item2} $\sum_n\gamma_n=+\infty$, $\sum_n\gamma_n^{1+\lambda}<\infty$ for some
    $\lambda \in (0,1)$,
  \item  \label{hyp:step:item3} $\sum_n |\gamma_n-\gamma_{n-1}|<\infty$\,.
  \end{enumerate}
\end{assum}
Polynomially decreasing sequences $\gamma_n \sim \gamma_\star / n^a$ when $n
\to \infty$, for some $a \in (1/2, 1]$ and $\gamma_\star>0$ satisfy Assumption~\ref{hyp:step}.
Finally, we introduce a stability-like condition.
\begin{assum}
  \label{hyp:stability}
  Almost surely, there exists a compact set $\cK$ of $\bR^{dN}$ such that
  $\theta_n \in \cK$ for any $n \geq 0$.
\end{assum}
Assumption~\ref{hyp:stability} claims that the sequence $(\theta_n)_{n \geq
  0}$ remains in a compact set and this compact set may depend on the path. It
is implied by the stronger assumption ``there exists a compact set $\cK$ of
$\bR^{dN}$ such that with probability one, $\theta_n \in \cK$ for any $n \geq
0$''.  Checking Assumption~\ref{hyp:stability} is not always an easy task. As
the main scope of this paper is the analysis of convergence rather than
stability, it is taken for granted: we refer to~\cite{bianchi:fort:hachem:2012}
for sufficient conditions implying stability.

\subsection{Results}
\label{sec:res_do}

The following lemma follows from standard algebra.
\begin{lemma}
  Under Assumptions~\ref{hyp:Wiid}-\ref{hyp:Wiid:item2}) and
  \ref{hyp:Wiid}-\ref{hyp:Wiid:item3}), the $N\times 1$ vector $v$ defined
  by $v^T \eqdef \frac 1N \un^T\overline W (I_N-J_\bot\overline W)^{-1}$ 
  is the unique non-negative vector satisfying $v^T=v^T\overline W$ and
  $v^T\un=1$.
\label{lem:irred}
\end{lemma}
% \textcolor{gray}{\begin{proof}
%   By the Jensen's inequality, for any $x \in \bR^N$, $x^T \overline{W}^T J_\bot
%   \overline{W} x \leq x^T \bE\left[ W_1^T J_\bot W_1\right]x$. Then, by
%   Assumption~\ref{hyp:Wiid}-\ref{hyp:Wiid:item3}), the spectral norm of
%   $J_\bot\overline W$ is strictly lower than one.  Therefore,
%   $I_N-J_\bot\overline W$ is
%   invertible.  \\
%   The vector $v$ satisfies $v^T \un =1$ and $v^T \overline{W} = v^T$; to that
%   goal, observe that $(I_N-J_\bot\overline W)^{-1} \un = \un$.  Let us prove
%   that a vector satisfying these two properties is unique; let $w \in \bR^N$
%   satisfying these properties. Then, $w^T = w^T \overline{W} = w^T J_\bot
%   \overline{W}+ \un^T
%   \overline{W}/N$ thus implying that $w^T = v^T$. \\
%   Since $\overline W$ is a stochastic matrix, its spectral radius is one.  By
%   \cite{horn:johnson:1990}, there exists a non-negative vector $w$ such that
%   $w^T\overline W=w^T$ and $\un^T w >0$.  We can therefore assume without loss
%   of generality that $w^T \un = 1$.  The above discussion implies that $w = v$.
%   This concludes the proof.
% \end{proof}}

If $A$ is a set, we say that $(x_n)_n$
converges to $A$ if $\inf\{|x_n-y| : y\in A\}$ tends to zero as $n\to\infty$.
\begin{theo} 
  Let Assumptions~\ref{hyp:distOpt},~\ref{hyp:Wiid},~\ref{hyp:step} and ~\ref{hyp:stability}  hold true.
Define the function $V: \bR^d \to \bR$
\begin{equation}
V(\theta) \eqdef \sum_{i=1}^N v_i\,f_i(\theta)\,\label{eq:VdistOpt}
\end{equation}
where $v = (v_1, \dots, v_N)$ is the vector defined in Lemma~\ref{lem:irred}.
Assume that the set $\cL = \{\theta \in \bR^d \, | \, \nabla V=0\}$ of critical
points of $V$ is non-empty and included in some level set $\{\theta: V(\theta)
\leq C \}$, and that $V(\cL)$ has an empty interior.  Assume also that the level
sets $\{\theta: V(\theta) \leq C \}$ are either empty or compact.  The
following holds with probability one:
\begin{enumerate}
\item \label{theo:distOpt:item1} The algorithm converges to a consensus i.e., $\lim_{n\to\infty} \max_{i,j} |\theta_{n,i}-\theta_{n,j}| = 0$.
\item \label{theo:distOpt:item2} The sequence $(\theta_{n,1})_{n \geq 0}$ converges to $\cL$ as $n \to
  \infty$.
\end{enumerate}
\label{the:distOpt}
\end{theo}
Theorem~\ref{the:distOpt} is proved in Appendix~\ref{sec:proofCV-distOpt}. Its
proof consists in showing that it is a special case of the more general
convergence result given by Theorem~\ref{the:cv}.

\subsection{Success and Failure of Convergence}
\label{sec:discussion}

% Theorem~\ref{the:distOpt} implies that
% the Algorithm~(\ref{eq:tempupdate})-(\ref{eq:gossipIntro}) generally fails to
% converge towards a critical point of the problem~(\ref{eq:distOpt}). Instead,
The algorithm converges to $\cL$ which in general is not the set of the
critical points of $\theta \mapsto \sum_i f_i(\theta)$. We discuss some
special where both sets actually coincide.

{\em {\bf \em Scenario 1.} All functions $f_i$ are strictly convex and admit a
  (unique) common minimizer $\theta_\star$.}

This case is for instance investigated by \cite{chen-sayed-sp12} in the
framework of statistical estimation in wireless sensor network. 
% \virer{In this scenario, we may assume without loss of
% generality that $f_i(\theta) \geq f_i(\theta_\star)=0$ for all $i$ (note that
% the Algorithm~(\ref{eq:tempupdate})-(\ref{eq:gossipIntro}) is not modified when
% $f_i$ is translated).  Since $v_i \geq 0$, $V$ is a non-negative strictly convex
% function such that $V(\theta_\star) =0$. Therefore, the set of minimizers of
% $V$ is $\{\theta_\star \}$. On the other hand, since $V$ is convex, $\cL$ is
% the set of minimizers of $V$. This implies that} 
The set {\sl $\cL$ is formed by
  the
  minimizers of $\, \sum_i f_i$}. 
Relaxing strict convexity, note that when the functions $f_i$ are just convex with a common minimizer and
$v_i>0$ for any $i$, then $\cL$ is formed by the minimizers of $\sum_i f_i$,
then the same conclusion holds.

{\em {\bf \em Scenario 2.} $\overline W$ is column-stochastic i.e.,
  $\un^T\overline W=\un^T$\,.}

In this case, $v$ given by Lemma~\ref{lem:irred} is the vector $\frac 1N\un$.
Consequently, $V = \frac 1N\sum_i f_i$.  Here again, $\cL$ is the set of
minimizers of $\sum_i f_i$. An example of random communication protocol (see \cite{aysal:2009})
satisfying $\un^T\overline W=\un^T$ is the following: at time $n$, a single
node $i$ wakes up at random with probability $p_i$ and broadcasts its temporary
update $\tilde \theta_{n,i}$ to all its neighbors $\cN_i$. Any neighbor $j$
computes the weighted average $\theta_{n,j}= \beta\tilde\theta_{n,i} +
(1-\beta) \tilde\theta_{n,j}$. On the other hand, any node $k$ which does not
belong to the neighborhood of $i$ (including $i$ itself) sets
$\theta_{n,k}=\tilde\theta_{n,k}$.
% Note that, as opposed to the pairwise
% scheme, the transmitter node $i$ does not expect any feedback from its neighbors.  
Then, given $i$ wakes up, the $(k,\ell)$th entry of $W_n$ is given by:
$$
w_n(k,\ell) = \left\{
  \begin{array}[h]{ll}
    1 & \textrm{if } k\notin \cN_i \textrm{ and }k=\ell \ ,\\
    \beta  & \textrm{if } k\in\cN_i \textrm{ and }\ell=i \ ,\\
    1-\beta & \textrm{if } k\in\cN_i \textrm{ and }k=\ell\ , \\
    0 & \textrm{otherwise.}
  \end{array}\right.
$$
Here, $W_n$ is {\em not} doubly stochastic. However, when nodes wake up according
to the uniform distribution ($p_i=\frac 1N$ for all $i$) it is easily seen that
$\un^T\bE[W_n]= \un^T$.

\subsection{Enhanced Algorithm with Weighted Step Sizes}
%\begin{remark}
\label{rem:weighted}
We end up this section with a simple modification of the initial algorithm in
the case where $v_i>0$ for all $i$.  Let us replace the local
  step~(\ref{eq:tempupdate}) of the algorithm by
  \begin{equation}
    \tilde \theta_{n,i} \eqdef \theta_{n-1,i} + \gamma_{n}\, v_i^{-1}Y_{n,i}
    \label{eq:weighted}
  \end{equation}
  where $Y_{n,i}$ is still given by~(\ref{eq:YnDistOpt}).
  As an immediate Corollary of Theorem~\ref{the:distOpt}, the
   algorithm~(\ref{eq:weighted})-(\ref{eq:gossipIntro}) drives the agent to a consensus which coincides with the critical points of $\sum_if_i$.

Of course, this modification requires for each node $i$ to have some prior knowledge of the communication protocol through the coefficients $v_i$
(in that case, questions related to a distributed computation of the $v_i$'s would be of interest, but are beyond the scope of this paper).
%\end{remark}

\section{Distributed Robbins-Monro Algorithm: General Setting}
\label{sec:general}
In this section, we consider the general setting described by
Algorithm~(\ref{eq:tempupdate})-(\ref{eq:gossipIntro}) with weaker conditions on
the distribution of the observations $Y_n$. We also weaken the assumptions on
 $(Y_{n+1}, W_{n+1})$: our general framework includes the case when the
communication protocol is adapted at each time $n$. 
%Although this has been recently proposed by \cite{rabbat:tsianos:2013} for distributed optimization, the gossip protocol is assumed to be double stochastic and the convergence is only shown by simulations.

We denote by $\mathsf M_1$ the set of $N\times N$
non-negative row-stochastic matrices and we endow $\mathsf M_1$ with its Borel
$\sigma$-field.
\begin{assum}
  \label{hyp:mutheta}
  \begin{enumerate}
  \item \label{hyp:mutheta:item1} There exists a collection of distributions
    $\left(\mu_\theta\right)_{\theta\in \bR^{dN}}$ on $\bR^{dN}\times \mathsf
    M_1$ such that a.s. for any Borel set $A$:
    $$
    \bP\left[ (Y_{n+1},W_{n+1}) \in A \, |\, \cF_n\right] =
    \mu_{\theta_n}(A)\ .
$$
In addition, the application $\theta \mapsto \mu_\theta(A)$ defined on $\bR^{dN}$ is measurable for any $A$ in the
Borel $\sigma$-field of $\bR^{dN}\times \mathsf M_1$. 
\item \label{hyp:mutheta:item2} For any compact set $\cK\subset \bR^{dN}$, $\sup\limits_{\theta\in \cK}\int
  |y|^2d\mu_\theta(y,w)<\infty$. 
\end{enumerate}
\end{assum}

Assumption~\ref{hyp:mutheta}-\ref{hyp:mutheta:item1}) means that the joint
distribution of the r.v.'s $Y_{n+1}$ and $W_{n+1}$ depends on the past
$\cF_{n}$ only through the last value $ \theta_{n}$ of the vector of estimates.
It also implies that $W_n$ is almost-surely (a.s.)  non-negative and
row-stochastic. Since the variables $(Y_{n+1}, W_{n+1})$ are not necessarily
independent conditionally to the past $\F_n$ and $(W_n)_{n \geq 1}$ are no
longer i.i.d., the contraction condition on $J_\bot W_1$ is replaced with the
following condition:
\begin{assum}
  \label{hyp:model}  For any compact set $\cK \subset \bR^{dN}$, there exists $\rho_\cK \in (0,1)$ such that for all $\theta \in \cK$, $\phi$
  in~$\mathbb{R}^{dN}$ and $A \in \bR^{dN} \times \bR^{dN}$,
$$
  \int |((J_\bot w) \otimes I_d) (\phi+ Ay)|^2
  d\mu_{\theta}(y, w)   \leq \rho_{\cK} \!\! \int\!\! |\phi+A y|^2 \, d\mu_\theta(y,w) \,.
$$
% \begin{multline*}
%   \int (\phi+ A y)^T (w \otimes I_d)^T \cJ_\bot (w \otimes I_d) (\phi+ Ay) \ 
%   d\mu_{\theta}(y, w)  \\
%   \leq \rho_{\cK} \, \int |\phi+A y|^2 \, d\mu_\theta(y,w) \,.
% \end{multline*}
\end{assum}
%We provide some insight on the above condition.  
Assumption~\ref{hyp:model} is
satisfied as soon as the spectral radius $ r\left(\bE\left[W_{1}^T \Jo W_{1}
    \vert \theta_0, Y_1 \right] \right)$ is upper bounded by a constant
independent of $(\theta_0,Y_1)$ when $\theta_0 \in \cK$ and strictly lower than
one.
%It ensures that $\sup_n r(\bE[W_n\Jo W_n\,|Y_n,\cF_{n-1}])<1$. \pbnote{Voir si on ne peut pas simplement donner ça comme condition.}
When $(W_n)_{n\geq 1}$ is an i.i.d. sequence, independent of the sequence
$(Y_n)_{n\geq 1}$ and of $\theta_0$, the above condition reduces to $r(\bE[W_1^T \Jo
W_1])<1$.

%%%%%%%%%%%%%%%%%%%%%%%%%%%%%%%%%%%%%%%%%%%%%%%%%%%%%%%%
%%%%%%%%%%%%%%%%%%%%%%%%%%%%%%%%%%%%%%%%%%%%%%%%%%%%%%%%
%%%%%%%%%%%%%%%%%%%%%%%%%%%%%%%%%%%%%%%%%%%%%%%%%%%%%%%%
%%%%%%%%%%%%%%%%%%%%%%%%%%%%%%%%%%%%%%%%%%%%%%%%%%%%%%%%
%%%%%%%%%%%%%%%%%%%%%%%%%%%%%%%%%%%%%%%%%%%%%%%%%%%%%%%%
%%%%%%%%%%%%%%%%%%%%%%%%%%%%%%%%%%%%%%%%%%%%%%%%%%%%%%%%
%%%%%%%%%%%%%%%%%%%%%%%%%%%%%%%%%%%%%%%%%%%%%%%%%%%%%%%%
%%%%%%%%%%%%%%%%%%%%%%%%%%%%%%%%%%%%%%%%%%%%%%%%%%%%%%%%
%%%%%%%%%%%%%%%%%%%%%%%%%%%%%%%%%%%%%%%%%%%%%%%%%%%%%%%%
%%%%%%%%%%%%%%%%%%%%%%%%%%%%%%%%%%%%%%%%%%%%%%%%%%%%%%%%
%%%%%%%%%%%%%%%%%%%%%%%%%%%%%%%%%%%%%%%%%%%%%%%%%%%%%%%%
%%%%%%%%%%%%%%%%%%%%%%%%%%%%%%%%%%%%%%%%%%%%%%%%%%%%%%%%
\section{Convergence Analysis}
\label{sec:CVanalysis}
For any vector ${x} \in\bR^{dN}$ of the form $x=(x_1^T,\dots,x_N^T)^T$ where
$x_i \in \bR^d$, we define the vector of $\bR^d$
$\la x\ra \eqdef (x_1+\cdots+x_N)/N =  ({\un^T\otimes I_d})x/N$.
We extend the notation to matrices $X\in \bR^{dN\times k}$ as $\la X\ra = \frac
1N ({\un^T\otimes I_d})X$. We note $\cJ \eqdef J \otimes I_d$ and $\cJ_\bot \eqdef J_\bot \otimes I_d$. 
Note that $\cJ x = \bs 1 \otimes \la x\ra$\,.
Algorithm~(\ref{eq:tempupdate}-\ref{eq:gossipIntro}) can be written in matrix
form as:
\begin{equation}
  \theta_n = \mathcal W_n\left(\theta_{n-1}+\gamma_{n}Y_n\right) \quad \text{where} \quad \mathcal W_n = W_n\otimes I_d \, .
  \label{eq:algo} 
\end{equation} We  decompose the estimate vector $\theta_n$ into two components
$\theta_n = \bs 1 \otimes \la \theta_n \ra + \cJ_\bot \theta_n$. In
Section~\ref{sec:disagreement}, we analyze the asymptotic behavior of the
disagreement vector $\cJ_\bot \theta_n$.  The study of the average vector $\la
\theta_n \ra$ will be addressed in Section~\ref{sec:average}. These two
sections are prefaced by a result which established the dynamics of these
sequences. Set $\alpha_n \eqdef \gamma_n / \gamma_{n+1}$ and
\begin{eqnarray}
  \label{eq:definition:phin}
&  \phi_n \eqdef  \gamma_{n+1}^{-1} \, \cJ_\bot\theta_n  \,.
%\\ &\alpha_n \eqdef \gamma_n / \gamma_{n+1} \, . \label{eq:definition:alphan}
\end{eqnarray}
The following lemma is left to the reader.
\begin{lemma}
  For each $n$, let $\theta_n$ be given by (\ref{eq:algo}) and let $W_n$
be row stochastic. Then,
  \begin{align}
    \athn &= \athnmu + \gamma_n \la \cW_n(Y_n+\phi_{n-1})\ra\ , \label{eq:SAmarkovnoise}\\
    \phi_n &= \alpha_n\, \cJ_\bot\cW_n(\phi_{n-1} + Y_{n})\label{eq:thOrth} \, .
  \end{align}
\end{lemma}
% \begin{proof}
%   Since $W_n$ is row-stochastic, $\cJ\cW_n\cJ=\cJ$. Hence, $\cJ\theta_n =
%   \cJ\theta_{n-1}+\cJ\cW_n\cJ_\bot\theta_{n-1}+\gamma_n\cJ\cW_nY_n$.  It
%   follows that $\cJ\theta_n =
%   \cJ\theta_{n-1}+\gamma_n\cJ\cW_n(Y_n+\gamma_n^{-1}\cJ_\bot\theta_{n-1})$
%   which directly gives~(\ref{eq:SAmarkovnoise}).  By projecting $\thn$ onto the
%   disagreement subspace, one has $\cJ_\bot\theta_n = \cJ_\bot\cW_n(\thnmu +
%   \gamma_n Y_{n})$. Since $W_n$ is row-stochastic,
%   $\cJ_\bot\cW_n=\cJ_\bot\cW_n\cJ_\bot$. Then, (\ref{eq:thOrth})
%   follows.
% \end{proof}

\subsection{Disagreement Vector}
\label{sec:disagreement}

\begin{lemma}
  \label{lem:agreement}
  Let Assumptions~\ref{hyp:step}-\ref{hyp:step:item1}), ~\ref{hyp:mutheta}
  and~\ref{hyp:model} hold.  Let $(\phi_n)_{n \geq 0}$ be the sequence given by
  (\ref{eq:definition:phin}).  For any compact set $\cK\subset \bR^{dN}$,
$\sup_n \bE \left(|\phi_n|^2 \indic_{\bigcap_{j \leq n-1} \{\theta_j \in \cK
  \}} \right)<\infty$.
\end{lemma}
The result is proved in Appendix~\ref{app:proof:sectiondisagreement}.
This lemma implies that for any compact set, there exists $C$ such that for any
$n \geq 0$, $\E[|\cJ_\bot\theta_n|^2\indic_{\bigcap_k \{\theta_k \in \cK_m \} }]\leq C
\gamma_{n+1}^2$.
\begin{prop}[Agreement] \label{prop:agreement}
  Under Assumptions~\ref{hyp:step}-\ref{hyp:step:item1}),
~\ref{hyp:step}-\ref{hyp:step:item2}), ~\ref{hyp:stability}, ~\ref{hyp:mutheta}
and~\ref{hyp:model}, 
$
\lim_{n\to\infty} \cJ_\bot\theta_n = 0
$ a.s.
\end{prop}
\begin{proof}
  Let $(\cK_m)_{m \geq 0}$ be an increasing sequence of compact subsets of
  $\bR^{dN}$ such that $\bigcup_m \cK_m = \bR^{dN}$.  Under
  Assumption~\ref{hyp:stability}, we have to prove equivalently that for any $m
  \geq 0$, $\lim_n \cJ_\bot\theta_n \un_{\bigcap_k \{\theta_k \in \cK_m \}} =0$
  a.s.  Let $m \geq 0$.  Lemma~\ref{lem:agreement} implies that there exists a
  constant $C$ such that for any $n$, $\E[|\cJ_\bot\theta_n|^2\indic_{\bigcap_k
    \{\theta_k \in \cK_m \} }]\leq C \gamma_{n+1}^2$.  By
  Assumption~\ref{hyp:step}-\ref{hyp:step:item2}), this implies that
  $\sum_n\E[|\cJ_\bot\theta_n|^2\indic_{\bigcap_k \{\theta_k \in \cK_m \}}]$ is
  finite; hence $ \sum_n |\cJ_\bot\theta_n|^2\indic_{\bigcap_k \{\theta_k \in
    \cK_m \}}$ is finite a.s. which yields $\lim_n
  \cJ_\bot\theta_n^2\indic_{\bigcap_k \{\theta_k \in \cK_m \}} = 0$ a.s.
\end{proof}

\subsection{Average vector}
\label{sec:average}
We now study the long-time behavior of the average estimate $\la \theta_n \ra$.
% To that goal, we write the update rule (\ref{eq:SAmarkovnoise}) as a stochastic
% approximation algorithm.  Let us explain the intuition for the definition of
% the associated mean field function.
Define for any $\theta\in \bR^{dN}$: \newcommand{\W}{\overline{\mathcal W}}
\begin{eqnarray}
  \W_\bth &\eqdef& \int  \left( w \otimes I_d \right) \,d\mu_{\bth}(y,w) \label{eq:Omega}\\
  z_\bth &\eqdef& \int ( w\otimes I_d)y\,d\mu_{\bth}(y,w)\, . \label{eq:upsilon} 
%%  m^{(1)}_\bth &\eqdef& (I_{dN}-\cJo\W_\bth)^{-1}\cJo z_\bth\,.
\end{eqnarray}
and let us assume regularity-in-$\theta$ properties of these quantities
\begin{assum}
\label{hyp:lipschitz} There exists $\lambda_\mu \in (1/2,1]$ and for any compact set $\cK\subset \bR^{dN}$, there exists a constant $C>0$
such that for any $\theta, \theta' \in\cK$,
\begin{align}
  & \left\| \W_\bth - \W_{\bth'}\right\|\leq C |\bth-\bth'|^{\lambda_\mu} \, ,   \label{eq:lipschitz2}  \\
& \left| \cJ z_\bth - \cJ z_{\cJ \bth} \right| \leq C |\cJo \bth|^{\lambda_\mu} \, , 
      \label{eq:lipschitz1}  \\
& \left| \cJ_\bot z_\bth - \cJ_\bot z_{\bth'} \right| \leq C |\bth -\bth'|^{\lambda_\mu} \, , 
      \label{eq:lipschitz3}
\end{align}
\end{assum}
We define the 
mean field function $h:\bR^d\to\bR^d$ (\ref{eq:thOrth}) by
\begin{equation}
  \label{eq:meanfield}
   h(\thetash) =  \la z_{\un\otimes\thetash} + \W_{\un\otimes\thetash} \, m_{\un \otimes\thetash}^{(1)}\ra
\end{equation}
where $m_{\un \otimes\thetash}^{(1)}$ is the expectation of the invariant
distribution $\pi_{1,\un \otimes \thetash}$, given by (see
Proposition~\ref{prop:momentPITheta} in Appendix~\ref{app:Phi})
\[
m^{(1)}_\bth \eqdef (I_{dN}-\cJo\W_\bth)^{-1}\cJo z_\bth\,.
 \]
 Note that under Assumption~\ref{hyp:model}, this quantity is well defined
 since for any compact $\cK \subset \bR^{dN}$, $\sup_{\theta \in \cK}
 \|\cJ_\bot \W_\bth\| \leq \sqrt{\rho_\cK}$.
\begin{assum} \label{hyp:lyapunov}
\begin{enumerate}
\item  \label{hyp:lyapunov:cont} $h: \bR^d \to \bR^d$ is continuous.
\item \label{hyp:lyapunov:V} There exists a continuously differentiable
  function $V: \bR^d \to \bR^+$ such that 
  \begin{enumerate}
  \item there exists $M>0$ such that $\cL = \{\thetash \in \bR^d: \nabla
    V^T(\thetash) h(\thetash) = 0 \} \subset \{V \leq M \}$. In addition,
    $V(\cL)$ has an empty interior;
  \item there exists $M'>M$ such that $\{V \leq M' \}$ is a compact subset of
    $\bR^d$;
  \item for any $\thetash \in \bR^d \setminus \cL$, $\nabla V^T(\thetash)
    h(\thetash) < 0$.
  \end{enumerate}
\end{enumerate}
\end{assum}
Assumptions~\ref{hyp:mutheta}, \ref{hyp:model} and \ref{hyp:lipschitz} imply
that $\thetash \mapsto m^{(1)}_{\un \otimes\thetash}$ is continuous on $\bR^d$
(see Proposition~\ref{prop:Regularity} in Appendix~\ref{app:Phi}). Therefore, a
sufficient condition for the
Assumption~\ref{hyp:lyapunov}-\ref{hyp:lyapunov:cont}) is to strengthen the
conditions (\ref{eq:lipschitz1}-\ref{eq:lipschitz3}) of
Assumption~\ref{hyp:lipschitz} as follows:
$
\left| z_\bth - z_{\bth'}\right| \leq C | \bth - \bth'|^{\lambda_\mu} $.
%  Observe that when $V$ is a continuous coercive function \emph{i.e.}, a continuous function such that
% $\lim_{|\thetash| \to \infty} V(\thetash) = \infty$, then the level sets $\{V
% \leq M \}$ are compact subsets of $\bR^{d}$.

\begin{prop}
  \label{prop:convergence:averagesequence} Let Assumptions~\ref{hyp:step}, \ref{hyp:stability}, \ref{hyp:mutheta}, \ref{hyp:model},\ref{hyp:lipschitz} and
  \ref{hyp:lyapunov} hold true.  Assume in addition that $\lambda \leq
  \lambda_\mu$ where $\lambda, \lambda_\mu$ are resp. given by
  Assumption~\ref{hyp:step} and \ref{hyp:lipschitz}. The average sequence $(\la
  \theta_{n} \ra)_n$ converges almost-surely to a connected component of $\cL$.
\end{prop}
%% \begin{prop}
%%   \label{prop:perturbedODE}
%%   Under Assumptions~\ref{hyp:mutheta}, \ref{hyp:model}, 
%% \ref{hyp:step}, \ref{hyp:stability} and~\ref{hyp:lipschitz}, there is a 
%% random sequence $(\eta_n)_{n\geq 0}$ such that for all $n$,
%% \begin{equation}
%%   \athn=\athnmu+\gamma_n h(\athnmu)+\gamma_n\eta_n 
%% \label{eq:perturbedODE}
%% \end{equation}
%% and $\lim_{N\to\infty}\sum_{n=0}^N\gamma_n\eta_n$ exists a.s.
%% \end{prop}
The proof of Proposition~\ref{prop:convergence:averagesequence} is given in
Appendix~\ref{sec:proof:cvg}.  It consists in verifying the assumptions of
\cite[Theorem 2]{andrieu:2005}.

\subsection{Main Convergence Result}
As a trivial consequence of Propositions~\ref{prop:agreement} and
\ref{prop:convergence:averagesequence}, we have
 \begin{theo}
   Let Assumptions~\ref{hyp:step}, \ref{hyp:stability}, \ref{hyp:mutheta},
   \ref{hyp:model}, \ref{hyp:lipschitz} and \ref{hyp:lyapunov} hold true.
   Assume in addition that $\lambda \leq \lambda_\mu$ where $\lambda,
   \lambda_\mu$ are resp. given by Assumption~\ref{hyp:step} and
   \ref{hyp:lipschitz}.  The following holds with probability one:
   \begin{enumerate}
   \item The algorithm converges to a consensus i.e., $\lim_{n\to\infty} \cJo\thn = 0$;
   \item $\theta_{n,1}$ converges to a connected
     component of $\cL$. 
   \end{enumerate}
 \label{the:cv}
 \end{theo}

%%%%%%%%%%%%%%%%%%%%%%%%%%%%%%%%%%%%%%%%%%%%%%%%%%%%%%%%
%%%%%%%%%%%%%%%%%%%%%%%%%%%%%%%%%%%%%%%%%%%%%%%%%%%%%%%%
%%%%%%%%%%%%%%%%%%%%%%%%%%%%%%%%%%%%%%%%%%%%%%%%%%%%%%%%
%%%%%%%%%%%%%%%%%%%%%%%%%%%%%%%%%%%%%%%%%%%%%%%%%%%%%%%%
%%%%%%%%%%%%%%%%%%%%%%%%%%%%%%%%%%%%%%%%%%%%%%%%%%%%%%%%
%%%%%%%%%%%%%%%%%%%%%%%%%%%%%%%%%%%%%%%%%%%%%%%%%%%%%%%%
%%%%%%%%%%%%%%%%%%%%%%%%%%%%%%%%%%%%%%%%%%%%%%%%%%%%%%%%
%%%%%%%%%%%%%%%%%%%%%%%%%%%%%%%%%%%%%%%%%%%%%%%%%%%%%%%%
%%%%%%%%%%%%%%%%%%%%%%%%%%%%%%%%%%%%%%%%%%%%%%%%%%%%%%%%
%%%%%%%%%%%%%%%%%%%%%%%%%%%%%%%%%%%%%%%%%%%%%%%%%%%%%%%%
%%%%%%%%%%%%%%%%%%%%%%%%%%%%%%%%%%%%%%%%%%%%%%%%%%%%%%%%
%%%%%%%%%%%%%%%%%%%%%%%%%%%%%%%%%%%%%%%%%%%%%%%%%%%%%%%%
%%%%%%%%%%%%%%%%%%%%%%%%%%%%%%%%%%%%%%%%%%%%%%%%%%%%%%%%
%%%%%%%%%%%%%%%%%%%%%%%%%%%%%%%%%%%%%%%%%%%%%%%%%%%%%%%%
%%%%%%%%%%%%%%%%%%%%%%%%%%%%%%%%%%%%%%%%%%%%%%%%%%%%%%%%
\section{Convergence Rate}
\label{sec:CvgRate}
\subsection{Main Result}
%\label{sec:CvgRate}
We derive the rate of convergence of the
sequence $\{\theta_n, n\geq 0\}$ to $\un\otimes \theta_\star $ for some $\theta_\star$ satisfying 
\begin{assum}\label{hyp:TCL:thetastar}
  $\theta_\star$ is a root of $h$ i.e., $h(\theta_\star) = 0$. Moreover, 
  $h$ is twice continuously differentiable in a neighborhood of
  $\theta_\star$.  The Jacobian $\nabla h(\theta_\star)$ is a Hurwitz matrix.
  Denote by $-L$, $L>0$, the largest real part of its eigenvalues.
\end{assum}
The moment conditions on the conditional distributions of the observations
$Y_n$ and the contraction assumption on the network have to be strengthened as
follows: 
\begin{assum}\label{hyp:TCL:momentY}
  There exists $\tau \in (0,2)$ such that for any compact set $\cK \subset \bR^{dN}$,
  one has $\sup_{\theta \in \cK} \int |y|^{2+\tau} \ d \mu_\theta(y,w) < \infty$.
\end{assum}
\begin{assum}\label{hyp:TCL:model}
  Let $\tau$ be given by Assumption~\ref{hyp:TCL:momentY}. For any compact set
  $\cK \subset \bR^{dN}$, there exists $\tilde \rho_\cK \in (0,1)$ such that
  for any $\phi \in \bR^{dN}$ 
\[
\sup_{\theta \in \cK} \int |((J_\bot w)\otimes I_d)|^{2 +\tau} d
\mu_\theta(y,w) \leq \tilde \rho_\cK \ |\phi|^{2+\tau} \ .
\]
\end{assum}
We also go further in the regularity-in-$\theta$ of the integrals w.r.t.
$\mu_\theta$. More precisely
\begin{assum}\label{hyp:TCL:ContinuiteMuTheta}
  There exists $\lambda_\mu \in (1/2,1]$ and for any compact set $\cK \subset
  \bR^{dN}$ there exists a constant $C$ such that
  \begin{enumerate}  
  \item \label{hyp:TCL:ContinuiteMuTheta:item2} for any $\theta, \theta' \in
    \cK$, $\left| \la z_\theta \ra - \la z_{\theta'} \ra \right| \leq C \left|\theta- \theta'
\right|^{\lambda_\mu}$.
\item \label{hyp:TCL:ContinuiteMuTheta:item2bis} Set $\mathcal{Q}_A(x,y,w)
  \eqdef (x+y)\T (w \otimes I_d)\T \cJ_\bot A \cJ_\bot (w \otimes I_d) (x+y)$
  for some $dN \times dN$ matrix $A$. For any $\theta, \theta' \in \cK$, $x \in
  \bR^{dN}$ and any matrix $A$ such that $\|A\| \leq 1$,
  \begin{multline*}
    \left| \int \mathcal{Q}_A(x,y,w) d\mu_\theta(y,w) - \mathcal{Q}_A(x,y,w)
      d\mu_{\theta'}(y,w) \ \right| \\
    \leq C \, \left|\theta - \theta'\right|^{\lambda_\mu} \ (1+|x|^2) \ .
  \end{multline*}
 \end{enumerate}

\end{assum}
We finally have to strengthen the conditions on the step-size sequence.
\begin{assum}\label{hyp:TCL:pas}   Let $\tau$ (resp. $\lambda_\mu$) be given by Assumption~\ref{hyp:TCL:momentY} (resp. Assumption~\ref{hyp:TCL:ContinuiteMuTheta}).  As $n \to \infty$, $\gamma_n \sim \gamma_\star/n^a$ for some  $a \in ((1+\lambda_\mu)^{-1} \vee (1+\tau/2)^{-1}; 1]$ and $\gamma_\star >0$. In addition, if $a =1$ then
  $\gamma_\star > 1/(2L)$ where $L$ is given by
  Assumption~\ref{hyp:TCL:thetastar}.
\end{assum}

%\subsection{Result}
%\label{sec:CvgRate:theo}
Define $m^{(1)}_\star  \eqdef (I_{dN}- \cJo\W_{\un\otimes\theta_\star})^{-1} \cJo z_{\un\otimes\theta_\star}$ and
$ m^{(2)}_\star \eqdef \left(I_{d^2N^2}- \Phi_\star\right)^{-1}\zeta_\star$
where $z_\theta$ is defined in~(\ref{eq:upsilon}), where 
\begin{eqnarray*}
  \Phi_\star &\eqdef& \int T(w)\,d\mu_{\un\otimes\theta_\star}(y,w) \\
\zeta_\star &\eqdef& \int T(w) \vect\left(yy^T+2m^{(1)}_\star y^T \right)\,d\mu_{\un\otimes\theta_\star}(y,w)
\end{eqnarray*}
and where we used the notation $T(w)\eqdef ((\Jo w)\otimes I_d)\otimes ((\Jo
w)\otimes I_d)$. As will be seen in the proofs, $m^{(1)}_\star$ and $m^{(2)}_\star$ represent
the asymptotic first order moment and (vectorized) second order moment of the r.v. $\phi_n$ defined by~(\ref{eq:definition:phin}).
Define also $R_\star(w)\eqdef(w\otimes I_d)-\W_{\un\otimes\theta_\star}$
and $\upsilon_\star(y,w) \eqdef (w\otimes I_d)y- z_{\un\otimes\theta_\star}$.
Finally, define
\begin{align*}
  \A_\star & \eqdef \left(\frac{\un^T}N \otimes I_d\right) (I_{dN}+\W_{\un\otimes\theta_\star}(I_{dN}-\cJo\W_{\un\otimes\theta_\star})^{-1} \cJo) \\
  \cR_\star & \eqdef \int  (R_\star(w)\otimes R_\star(w))\ d\mu_{\un\otimes\theta_\star}(y, w) \\
  \cT_\star & \eqdef  \int (\upsilon_\star(y,w)\otimes R_\star(w)) d\mu_{\un\otimes\theta_\star}(y, w) \\
  \cS_\star & \eqdef  \int \vect (\upsilon_\star(y,w)\upsilon_\star(y,w)^T) d\mu_{\un\otimes\theta_\star}(y, w)\,.
\end{align*}
We establish in Section~\ref{sec:proof:TCL} the following result.
\begin{theo}\label{theo:TCL}
  Let Assumption~\ref{hyp:mutheta}-\ref{hyp:mutheta:item1}),
  Assumption~\ref{hyp:lipschitz}, Assumption~\ref{hyp:model} and
  Assumption~\ref{hyp:TCL:thetastar} to Assumption~\ref{hyp:TCL:pas} hold true.
  Let $U_\star$ be the positive-definite matrix given by
% \begin{multline*}
% U_\star = \int \int \left( \la w \left(y + x \right) \ra - \int \la
%   w'(y'+x) \ra d \mu_{\theta_\star \un}(y',w')   \right. \\
%  \left.  + \left( \int \la w'\ra d \mu_{\theta_\star \un}(y',w') \right)
%   \left( g_{1, \theta_\star \un}\left( J_\bot w (x +y)\right) - P_{1,
%       \theta_\star \un} g_{1, \theta_\star \un}(x) \right) \right)^2 \ 
% d\mu_{\theta_\star \un}(y, w) \ d\pi_{1,\theta_\star \un}(dx) \ .
% \end{multline*}
 $$
 \vect U_\star = (\A_\star\otimes \A_\star) (\cR_\star\,m^{(2)}_\star +2\cT_\star m^{(1)}_\star +\cS_\star)
 $$
Then conditionally to the event $\{\lim_n \theta_n = \un \otimes \theta_\star \}$, the
sequence $\{\gamma_n^{-1/2} (\la \theta_n \ra - \theta_\star), n\geq 0\}$ converges
in distribution to a zero mean  Gaussian distribution with
covariance matrix $\mathsf{V}$ where $\mathsf{V}$ is the unique positive-definite
matrix satisfying 
\vspace{-0.2cm}
\small

\begin{align*}
&\mathsf{V} \nabla h(\theta_\star)^T  +  \nabla h(\theta_\star) \mathsf{V} =  - U_\star \quad  \text{if $a<1$,} \\
&\mathsf{V} \left( I_{d} + 2 \gamma_\star\nabla h(\theta_\star)\right)^T  + \left( I_{d} + 2 \gamma_\star\nabla h(\theta_\star)\right)  \mathsf{V} =  - 2 \gamma_\star  U_\star \quad  \text{if $a=1$.}
\end{align*}

% \textcolor{red}{Il faut faire vivre ce théorème ... : le reformuler dans des
%   cas plus simples; expliquer ce que sont $\pi_{\alpha, \theta}$, $P_{\alpha,
%     \theta}$ et $g_{\alpha,\theta}$}
\end{theo}

\normalsize
\vspace{0.1cm}

\subsection{A Special Case: Doubly-Stochastic Matrices}
\label{sec:TCL:doublySto}
In this paragraph, let us investigate the special case when $(W_n)_n$ are
$N\times N$ doubly-stochastic matrices.  Note that in this case,
(\ref{eq:SAmarkovnoise}) gets into $\la \theta_n \ra = \la \theta_{n-1} \ra +
\gamma_n \la Y_n \ra$ and the mean field function $h$ is equal to
$h(\vartheta)=\int \la y \ra d\mu_{\un \otimes \vartheta}(y,w)$. 
Since $W_n$ is column-stochastic, $\int w\,d\mu_{\un\otimes\theta_\star}(y,w)$
is column-stochastic, and we have $\A_\star = \frac{\un^T}N \otimes I_d$.
Then, it is not difficult to check that $\A_\star R_\star(w) = 0$, which
implies that $\cR_\star=\cT_\star=0$.  
% Therefore $\vect U_\star =
% (\A_\star\otimes \A_\star) \cS_\star$ \emph{i.e.}, $U_\star = \int \la
% \upsilon_\star(y,w)\ra\la\upsilon_\star(y,w)\ra^T
% d\mu_{\un\otimes\theta_\star}(y, w)$. 
This yields the following corollary
\begin{coro}
\label{coro:double_case}
  In addition to the assumptions of Theorem~\ref{theo:TCL}, assume that
  $(W_n)_n$ are $N\times N$ doubly-stochastic matrices and set $\bar y_\star=\int y\, d\mu_{\un\otimes\theta_\star}(y, w) $. Then 
$$U_\star = \int \la y-\bar y_\star\ra\la y-\bar y_{\star}\ra^T d\mu_{\un\otimes\theta_\star}(y, w) \,.$$
\end{coro}

% \textcolor{red}{A recaser qq part: ``Loosely speaking, this holds when the
%   matrices $W_n$ tend to become column-stochastic \emph{in average} as $n$
%   tends to infinity. ``}

\section{Concluding remarks}
\label{sec:conc_remarks}

In this paragraph, we informally draw some general conclusions of our study.
We assimilate the communication protocol to the selection of the sequence $W_n$,
which we assume i.i.d. in this paragraph for simplicity.
We say that a protocol is doubly stochastic if $W_n$ is doubly stochastic for each $n$.
We say that a protocol is doubly stochastic \emph{in average} if $\EE{W_n}$ is doubly stochastic for each $n$.

\begin{enumerate}
\item {\bf Consensus is fast.}
Theorem~\ref{theo:TCL} states that the average estimation error
converges to zero at rate $\sqrt{\gamma_n}$.  This result was actually
expected, as $\sqrt{\gamma_n}$ is the well-known convergence rate of
standard stochastic approximation algorithms. 

On the other hand, Lemma~\ref{lem:agreement} suggests that the disagreement
vector $\cJo\theta_n$ goes to zero at rate $\gamma_n$ that is, one order of
magnitude faster.  Asymptotically, the fluctuations of the normalized
estimation error $(\theta_n-\un\otimes \theta_\star)/\sqrt{\gamma_n}$ are fully
supported by the consensus space.

This remark also suggests to analyze non-stationary communication protocols, for which 
the number of transmissions per unit of time decreases with $n$. This problem is addressed in~\cite{bianchi:fort:hachem:2012}.

\item {\bf Non-doubly stochastic protocols generally influence the limit points.} 
This issue is discussed in Section~\ref{sec:discussion}.
The choice of the matrices $W_n$ is likely to have an impact on the set of limit points of the algorithms.
This may be inconvenient especially in distributed optimization tasks.

\item {\bf Protocols that are doubly stochastic "in average'' all lead to the same limit points}. 
In the framework of distributed optimization, the latter set of limit points precisely coincides with the sought
critical points of the minimization problem. It means that non-doubly stochastic protocols can be used
provided that they are doubly stochastic in average.

\item {\bf Asymptotically, doubly stochastic protocols perform as well as a
    centralized algorithm.}  By Corollary~1, if $W_n$ is chosen to be doubly
  stochastic for all $n$, the asympotic error covariance characterized in
  Theorem~\ref{theo:TCL} does not depend on the specific choice of $W_n$. In distributed
  optimization, the asymptotic performance is identical to the performance that
  would have been obtained by replacing $W_n$ by the orthogonal projector $J$,
  which would lead to the centralized update
$
\la\theta_n\ra = \la\theta_{n-1}\ra + \frac{\gamma_n}N\sum_{i=1}^N  Y_{n,i}\,.
$
On the opposite, protocols that are not doubly stochastic generally influence the asymptotic error covariance, \emph{even if they are doubly stochastic in average}.
\end{enumerate}

\section{Numerical Results}
We illustrate the convergence results obtained in Section~\ref{sec:res_do} and discussed in sections~\ref{sec:discussion} and \ref{sec:conc_remarks}. We depict a particular case of the distributed optimization problem described in Section~\ref{sec:dist_opt}. Consider a network of $N=5$
agents and for any $i=1,\dots,5$, we define a private cost function $f_i:\bR\to\bR$.
We address the following minimization problem:
\begin{align}
\label{eq:ex}
\min_{\theta \subset \bR} \, \sum_{i=1}^5 \frac{1}{2}(\theta - \alpha_i)^2
\end{align}
where $\alpha^T=(-3,5,5,1,-3)$. The minimizer of~\eqref{eq:ex} is $\theta_f = \la \alpha \ra =1$. The network is represented by an undirected graph $G=(V,E)$ with vertices $\{1,\dots,N\}$ and $6$ fixed edges $E$. The corresponding adjacency matrix is given by
$$
A = \begin{pmatrix}
  0 & 1 & 0 & 1 & 0 \\
  1 & 0 & 1 & 0 & 0 \\
  0& 1 & 0 & 1 & 1  \\
1 & 0 & 1 & 0 & 1 \\
0 & 0 & 1 & 1 & 0
 \end{pmatrix}.
$$
We choose $\theta_{0,i} =0$ for each agent $i$ and the step-size sequence of the form
$\gamma_n = 0.1/n^{0.7}$. Observations $Y_{n,i} $ are defined as
in~\eqref{eq:YnDistOpt}: $(\xi_{n,i})_{n,i}$ is an i.i.d. sequence with Gaussian
distribution ${\cal N}(0,\sigma^2)$ where $\sigma^2=1$. 

Figure~\ref{fig:theo1} illustrates the two results of Theorem~\ref{the:distOpt} according to different gossip matrices $(W_n)_n$. First, Figure~\ref{subfig:convergence} addresses the convergence of sequence $(\theta_{n,1})_{n \geq 0}$ as a function of~$n$ to show the influence of matrices $W_n$ to the limit points. In particular, the dashed line curve corresponds to the algorithm~(\ref{eq:tempupdate})-(\ref{eq:gossipIntro}) when $W_n$ is assumed to be fixed and deterministic
($W_n=W_1$ for all $n$); we select $W_1$ in such a way that each agent computes the average
of the temporary estimates in its neighborhood. This is equivalent to set $W_1
= (I_N+D)^{-1}(I_N+A)$, where
$D$ is the diagonal matrix containing the degrees, \emph{i.e.} $D(i,i) = \sum_{j=1}^NA(i,j)$ for each agent $i$. Note that $W_1$ is not doubly stochastic since $\un^T
W_1\neq \un^T$. Computing the left Perron eigenvector defined by Lemma~\ref{lem:irred} yields the minimizer of $V= \sum_i v_if_i$ being $\theta_V =v^T\alpha=1.24$. In that case, the sequence $(\theta_{n,1})_{n}$
converges to $\theta_\star=\theta_V$ instead of the desired $\theta_\star=\theta_f$. Figure~\ref{subfig:convergence} includes the trajectory of $\theta_{n,1}$ generated by Algorithm~\eqref{eq:weighted}-\eqref{eq:gossipIntro} with $W_1
= (I_N+D)^{-1}(I_N+A)$. As proposed in Section~\ref{rem:weighted} when introducing the weighted step size such $\gamma_nv_i^{-1}$ the sequence now converge to the sought value $\theta_f$. 
%\begin{figure}[h!]
%        \centering
%        \begin{subfigure}{0.5\textwidth}
%                  \includegraphics[width=0.95\textwidth,height=4.8cm]{figures/failure.pdf}
%                \caption{Failure cases.} 
%                \label{subfig:failure}
%        \end{subfigure}%
%        \\
%        %add desired spacing between images, e. g. ~, \quad, \qquad, \hfill etc.
%          %(or a blank line to force the subfigure onto a new line)
%        \begin{subfigure}{0.5\textwidth}
%                 \includegraphics[width=0.95\textwidth,height=4.8cm]{figures/success.pdf}
%                \caption{Success cases.} 
%                \label{subfig:success}
%        \end{subfigure}
%        \caption{Trajectories of $\theta_{n,1}$ as a function of $n$ according to different communication schemes.}
%        \label{fig:convergence}
%\end{figure}

\begin{figure}[h!]
        \centering
        \begin{subfigure}{0.5\textwidth}
        \centering
                  \includegraphics[width=0.95\textwidth,height=5.5cm]{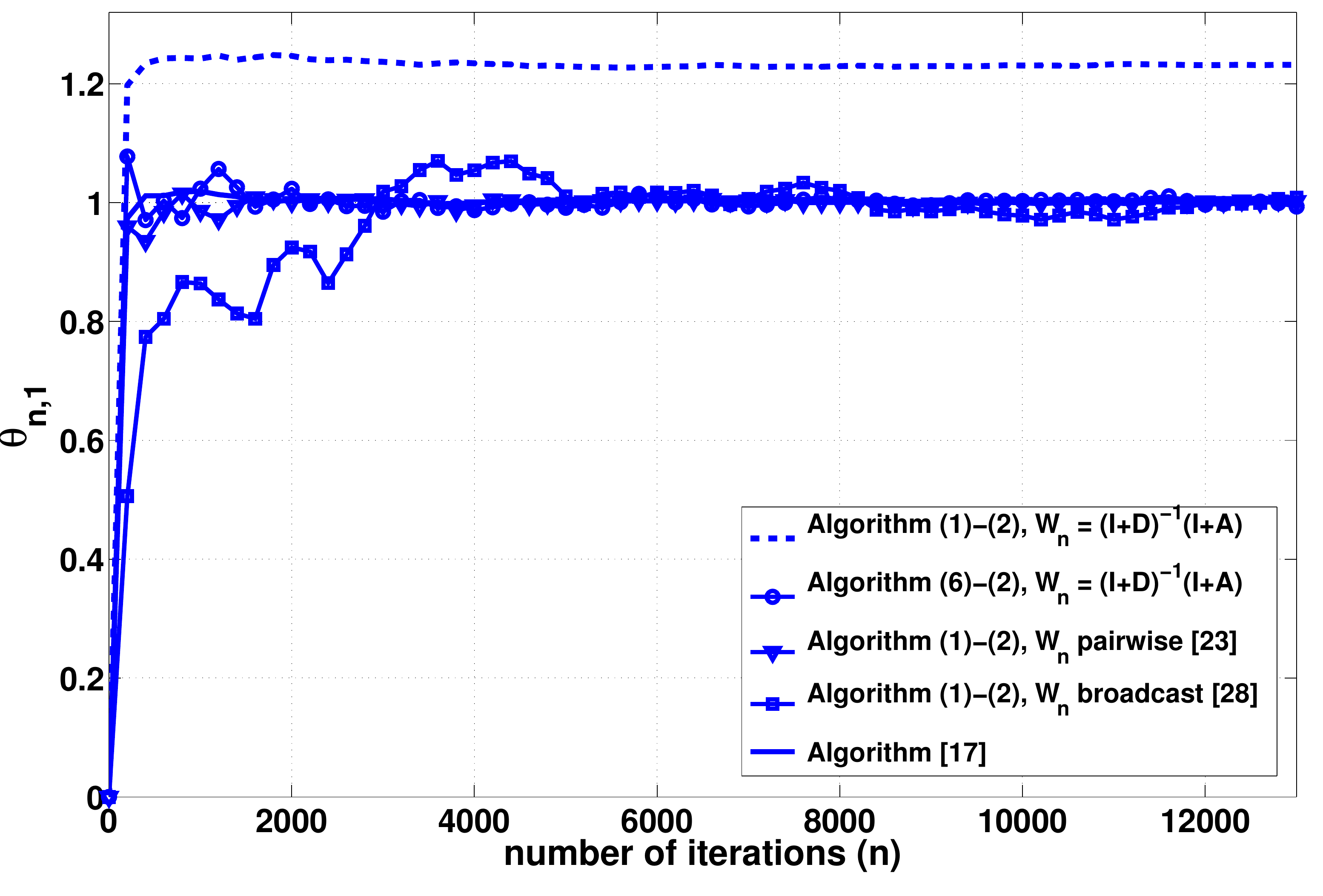}
                \caption{Trajectories of $\theta_{n,1}$ as a function of $n$.} 
                \label{subfig:convergence}
        \end{subfigure}%
        \\
        \vspace{0.1cm}
        \begin{subfigure}{0.5\textwidth}
        \centering
                 \includegraphics[width=0.95\textwidth,height=5.5cm]{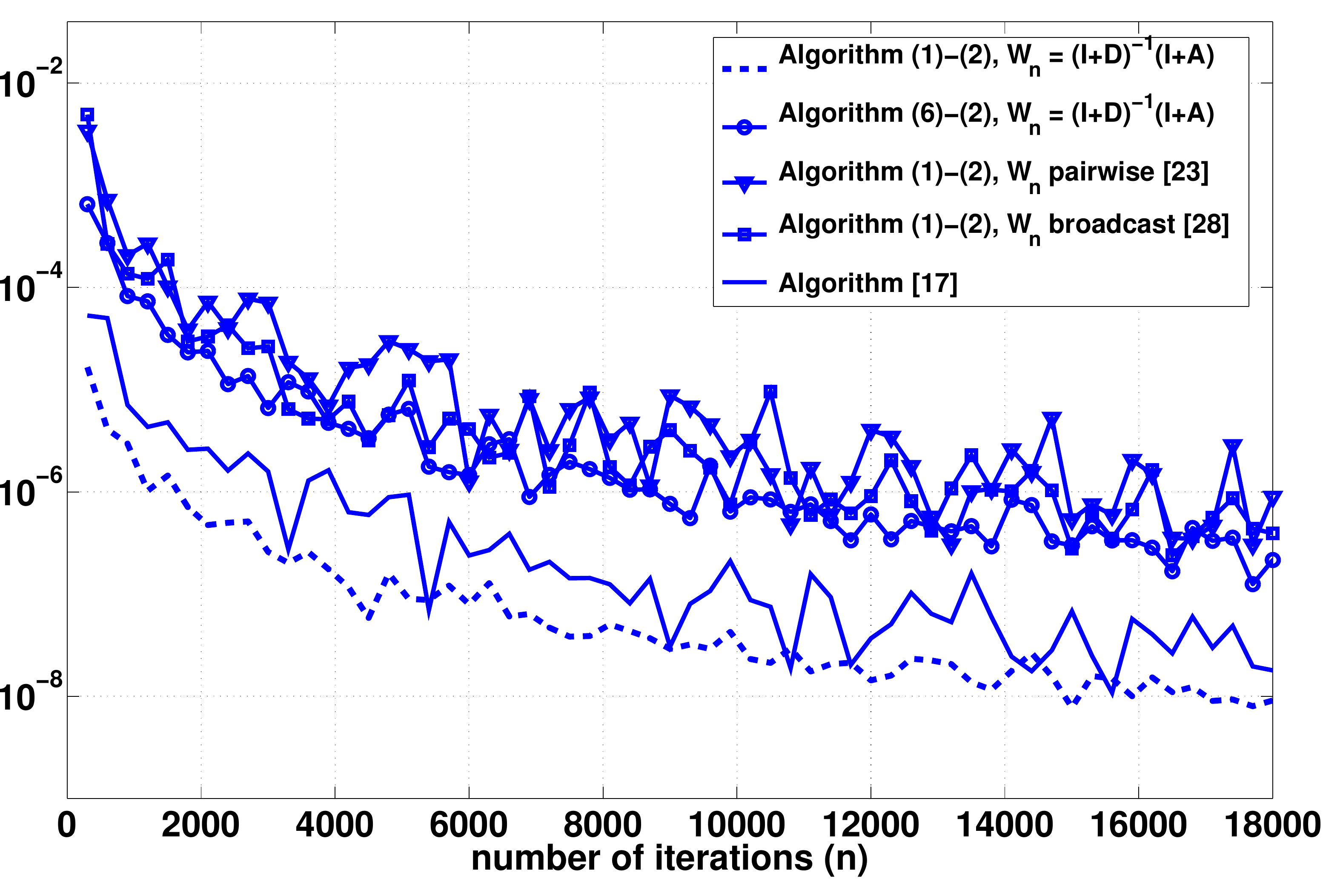}
                \caption{$\sqrt{\frac 1N\sum_{i=1}^N (\theta_{n,i} - \athn)^2}$ as a function of $n$.} 
                \label{subfig:consensus}
        \end{subfigure}
        \caption{Convergence result of Theorem~\ref{the:distOpt} according to different communication schemes for $(W_n)_n$.}
        \label{fig:theo1}
\end{figure}
Figure~\ref{subfig:convergence} also illustrates the convergence behavior of Scenario~2 where the limit point $\theta_\star$ of Algorithm~(\ref{eq:tempupdate})-(\ref{eq:gossipIntro}) corresponds with $\theta_f$. In that case, we consider two standard models for $W_n$, namely the pairwise gossip of~\cite{boyd:2006} and the broadcast gossip of~\cite{aysal:2009} (we set $\beta=\frac 12$). Finally, the plain line in Figure~\ref{subfig:convergence} shows the performance of the algorithm proposed by~\cite{nedic2013distributed} for distributed optimization which is based on a synchronous version of the push-sum model of \cite{kempe:dobra:gehrke:focs-2003}.

%\begin{figure}[h!]
%  \centering
%  \includegraphics[width=\columnwidth,height=4.8cm]{figures/comparison.pdf}
%  \caption{Trajectory of $\theta_{n,1}$ as a function of $n$.} 
%% when using constant matrix $W_n=W_1$.
%%  The two curves correspond to the local step (\ref{eq:tempupdate}) with standard step size $\gamma_n$ and weighted step size $\frac{\gamma_n}{v_1}$ respectively.
%  \label{fig:traj}
%\end{figure}

We conclude the illustration of Theorem~\ref{the:distOpt} by the results on the consensus convergence for the same examples of $W_n$ considered in Figure~\ref{subfig:convergence}. Thus, Figure~\ref{subfig:consensus} represents the norm of
the scaled disagreement vector as a function of~$n$. As expected from Theorem~\ref{the:distOpt}-\ref{theo:distOpt:item2}), consensus is asymptotically achieved independently of the limit point, \emph{i.e.} $\theta_f$ or $\theta_V$. Note that the synchronous models of $W_1$ and \cite{nedic2013distributed} require $N$ transmissions at each iteration $n$ whereas the gossip protocols of \cite{boyd:2006} and \cite{aysal:2009} only require two and one transmissions respectively due to their asynchronous nature. This may explain the gap between the curves in Figure~\ref{subfig:consensus} when regarding the convergence rate towards the consensus.
%\begin{figure}%[h!]
%  \centering
%  \includegraphics[width=0.97\columnwidth,height=4.8cm]{figures/mdc.pdf}
%  \caption{ $\sqrt{\frac 1N\sum_{i=1}^N (\theta_{n,i} - \athn)^2}$ as a function of $n$.}
%  \label{fig:consensus}
%\end{figure}

The result of Theorem~\ref{theo:TCL} is illustrated in Figure~\ref{fig:theo3} which leads to the concluding remark 4) of Section~\ref{sec:conc_remarks}. Figures~\ref{subfig:boxplot} and \ref{subfig:delta} display the asymptotic analysis of the normalized average error $ \gamma_n^{-1/2}(\la \theta_n\ra - \theta^\star)$. Indeed, once the convergence is achieved, the asymptotic distribution can be characterized by the closed form of the variance $U^\star \in \bR$. In this example, Theorem~\ref{theo:TCL} states that $ \gamma_n^{-1/2}(\la \theta_n\ra - \theta^\star)$ converges in distribution to a r.v. $\sim \mathcal N(0,\mathsf{V})$ where $\nabla h(\theta_\star)=-1$ and thus the variance is $\mathsf{V}=\frac{U_\star}{2}$. The first boxplot and the first histogram in Figure~\ref{fig:theo3} are related to the algorithm implemented in a centralized manner. We consider the distributed algorithm~(\ref{eq:tempupdate})-(\ref{eq:gossipIntro}) with different choices of $W_n$: the pairwise gossip of \cite{boyd:2006}, the broadcast gossip of \cite{aysal:2009} and the fixed $W_1$ defined by $(I_N+D)^{-1}(I_N+A)$. The normal distribution obtained in Theorem~\ref{theo:TCL} is coherent with the empirical results.

\begin{figure}[h!]
        \centering
%        \begin{subfigure}{0.5\textwidth}
%                  \includegraphics[width=0.95\textwidth,height=4.8cm]{figures/speed.pdf}
%                \caption{Trajectories of $\theta_{n,1}^\bot$ and $\la \theta_n\ra$ as a function of $n$.} 
%                \label{subfig:speed}
%        \end{subfigure}%
         \begin{subfigure}{0.5\textwidth}
         \centering
                  \includegraphics[width=0.84\textwidth,height=4.8cm]{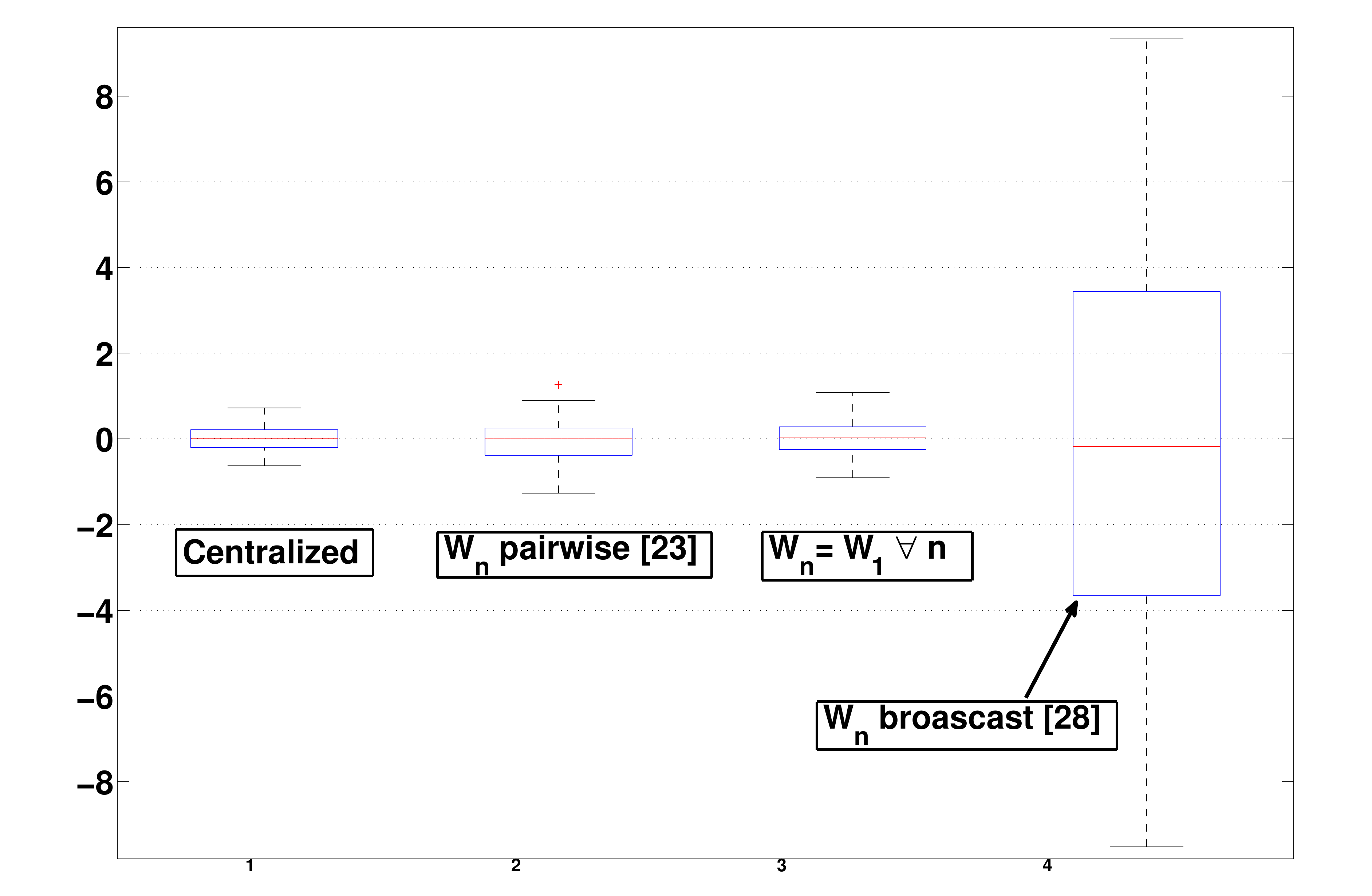}
                \caption{Boxplots of the normalized average error.} 
                \label{subfig:boxplot}
        \end{subfigure}%
        \\%add desired spacing between images, e. g. ~, \quad, \qquad, \hfill etc.
          \vspace{0.1cm}
        \begin{subfigure}{0.5\textwidth}
        \centering
                 \includegraphics[width=\textwidth,height=5.2cm]{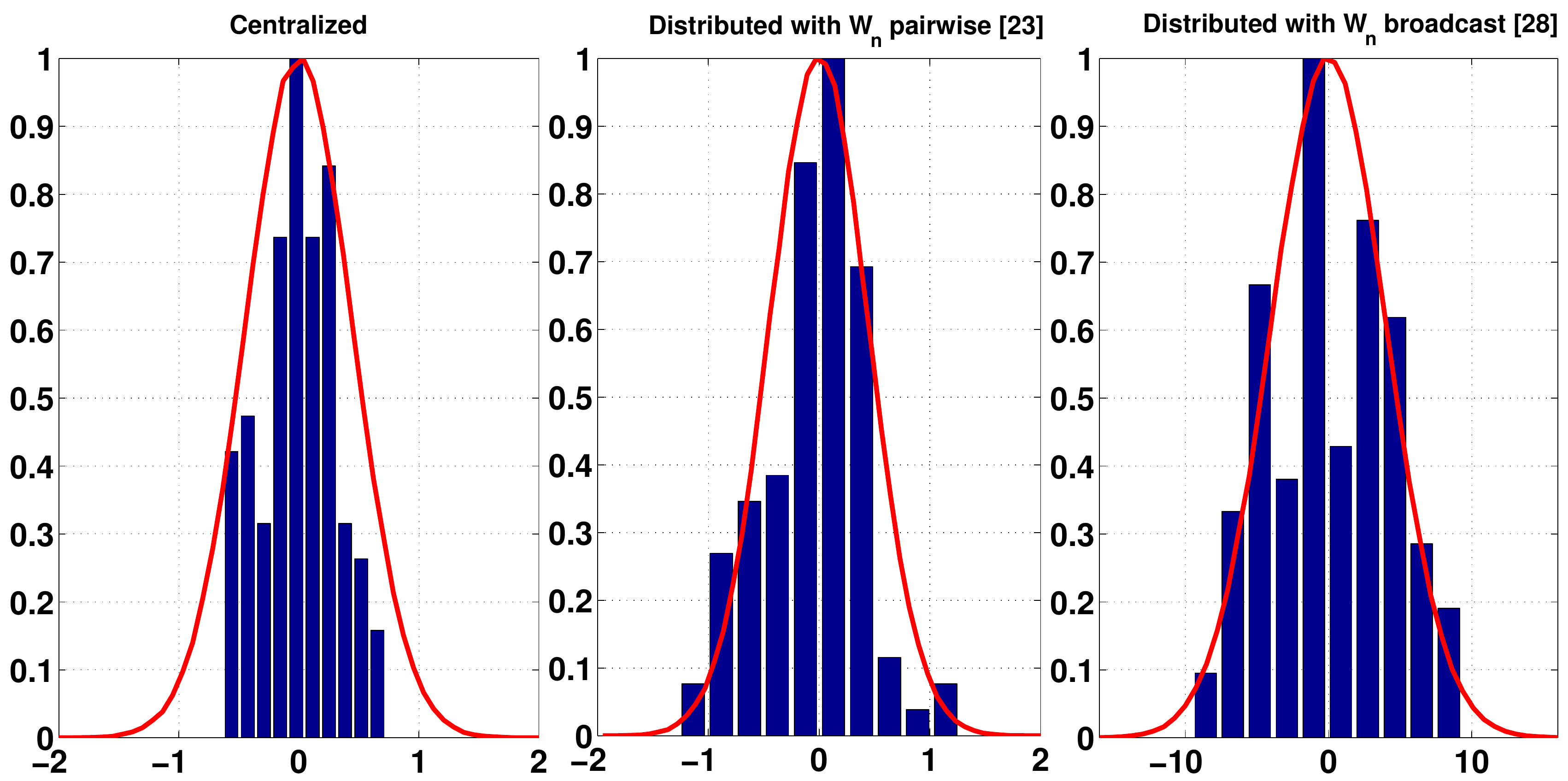}
                \caption{Empirical distribution (dark bars) \emph{versus} theoretical distribution given by Theorem~\ref{theo:TCL} (solid line).} 
                \label{subfig:delta}
        \end{subfigure}
        \caption{Asymptotic analysis of the normalized average error $\frac{1}{\sqrt{\gamma_n}}(\la \theta_n\ra - \theta^\star)$ of Algorithm~(\ref{eq:tempupdate})-(\ref{eq:gossipIntro}) according to different communication schemes for $(W_n)_n$ after $n=30000$ iterations and over $100$ independent Monte-Carlo runs.}
        \label{fig:theo3}
\end{figure}

%%%%%%%%%%%%%%%%%%%%%%%%%%%%%%%%%%%%%%%%%%%%%%%%%%%%%%%%
%%%%%%%%%%%%%%%%%%%%%%%%%%%%%%%%%%%%%%%%%%%%%%%%%%%%%%%%
%%%%%%%%%%%%%%%%%%%%%%%%%%%%%%%%%%%%%%%%%%%%%%%%%%%%%%%%
%%%%%%%%%%%%%%%%%%%%%%%%%%%%%%%%%%%%%%%%%%%%%%%%%%%%%%%%
%%%%%%%%%%%%%%%%%%%%%%%%%%%%%%%%%%%%%%%%%%%%%%%%%%%%%%%%
%%%%%%%%%%%%%%%%%%%%%%%%%%%%%%%%%%%%%%%%%%%%%%%%%%%%%%%%
%%%%%%%%%%%%%%%%%%%%%%%%%%%%%%%%%%%%%%%%%%%%%%%%%%%%%%%%
%%%%%%%%%%%%%%%%%%%%%%%%%%%%%%%%%%%%%%%%%%%%%%%%%%%%%%%%
%%%%%%%%%%%%%%%%%%%%%%%%%%%%%%%%%%%%%%%%%%%%%%%%%%%%%%%%
%%%%%%%%%%%%%%%%%%%%%%%%%%%%%%%%%%%%%%%%%%%%%%%%%%%%%%%%
%%%%%%%%%%%%%%%%%%%%%%%%%%%%%%%%%%%%%%%%%%%%%%%%%%%%%%%%
%%%%%%%%%%%%%%%%%%%%%%%%%%%%%%%%%%%%%%%%%%%%%%%%%%%%%%%%
\appendices

\section{Proof of Theorem~\ref{the:distOpt}}
\label{sec:proofCV-distOpt} 
We prove that the Assumptions~\ref{hyp:mutheta} to
  \ref{hyp:lyapunov} hold.  Then Theorem~\ref{the:distOpt} will follow from
  Theorem~\ref{the:cv}.  For any $\theta =(\theta_1, \dots, \theta_N) \in
  \bR^{dN}$ where $\theta_i \in \bR^d$, define the $\bR^{dN}$-valued function
  $g$ by
$
g(\theta) \eqdef (-\nabla f_1(\theta_1)^T,\dots,-\nabla f_N(\theta_N)^T)^T$.
Under Assumption~\ref{hyp:Wiid}-\ref{hyp:Wiid:item1}) and
Assumption~\ref{hyp:Wiid}-\ref{hyp:Wiid:item2}), for any Borel set $A
\times B$ of $\bR^{dN} \times \mathsf M_1$
$
\bP[(Y_{n+1},W_{n+1})\in A\times B \vert  \cF_n]=\bP[Y_{n+1}\in A \vert 
\cF_n]\bP[W_{n+1}\in B]$.
In addition, by Assumption~\ref{hyp:distOpt} and Eq. (\ref{eq:YnDistOpt})
$
\bP[Y_{n+1}\in A \vert \cF_n] = \int \indic_{A}\left(g(\theta_n) +z\right) \ 
d\nu_{\theta_n}(z) \, .
$
The above discussion provides the expression of $\mu_\theta$ in
Assumption~\ref{hyp:mutheta}-\ref{hyp:mutheta:item1}). In addition, under
Assumption~\ref{hyp:distOpt}-\ref{hyp:distOpt:item2}), for any compact set
$\cK$ of $\bR^{dN}$,
\small
\[
\sup_{\theta\in \cK}\int |y|^2d\mu_\theta(y,w)=\sup_{\theta\in
  \cK}\left(|g(\theta)|^2 + \int |z|^2d\nu_\theta(z)\right) < \infty
\]
\normalsize
which proves Assumption~\ref{hyp:mutheta}-\ref{hyp:mutheta:item2}). 
% Set $
% \cW=w\otimes I_d$.  The above expression of $\mu_\theta$ implies that
% \begin{multline*}
%    \int (\phi+ A y)^T \cW^T \cJ_\bot \cW (\phi+ Ay) \ 
%   d\mu_{\theta}(y, w) \\
% = \int  (\phi+ A (g(\theta) +z))^T  \bE\left[ \cW_1^T \cJ_\bot \cW_1 \right] (\phi+ A(g(\theta)+z)) \ 
%   d\nu_{\theta}(z)\, .
% \end{multline*}
% Therefore, 
Assumption~\ref{hyp:model} easily follows from
Assumption~\ref{hyp:Wiid}-\ref{hyp:Wiid:item3}).   The regularity conditions of
Assumption~\ref{hyp:lipschitz} are satisfied with $\lambda_\mu = \delta$, where
$\delta$ is given by Assumption~\ref{hyp:distOpt:item1}.
 Observe indeed that the
left hand side of \eqref{eq:lipschitz2} is zero and \eqref{eq:lipschitz1} and
\eqref{eq:lipschitz3} are true as long as $(\nabla f_i)_{i}$ are locally
Hölder-continuous. Again, the expression of $\mu_\theta$ implies that
$  \W_\bth =   \bE\left[\cW_1\right]$.
Therefore, 
$ h(\vartheta) = \la \bE[\cW_1] \, g\left(\un \otimes \vartheta\right)\ra = -\sum_{i=1}^N v_i \nabla f_i(\vartheta)$
which completes the proof.
% where $v = (v_1, \dots, v_N)$ is given by Lemma~\ref{lem:irred}.  Set $\bar V
% \eqdef \exp(V)$ where $V$ is defined by \eqref{eq:VdistOpt}.  Upon noting that
% $\nabla \bar V = -h \, \bar V$, it is easily seen that under the assumptions
% of Theorem~\ref{the:distOpt}, Assumption~\ref{hyp:lyapunov} holds.

\section{Proof of  Lemma~\ref{lem:agreement}}
\label{app:proof:sectiondisagreement}
From~(\ref{eq:thOrth}), we compute $|\phi_n|^2 = \alpha_n^2(\phi_{n-1}
+Y_{n})^T\cW_n^T\cJ_\bot\cW_n (\phi_{n-1}+Y_{n})$. Using
Assumption~\ref{hyp:mutheta}-\ref{hyp:mutheta:item1}),
$\E\left[|\phi_n|^2|\cF_{n-1} \right]$ is equal to
\small
\[
\alpha_n^2 \ \int (\phi_{n-1}+ y)^T (w \otimes I_d) \cJ_\bot (w \otimes I_d)
(\phi_{n-1}+y) \ d\mu_{\theta_{n-1}}(y, w) \, .
\]
\normalsize
By Fubini Theorem and Assumption~\ref{hyp:model}, there exists $\rho_\cK \in
(0,1)$ such that for any $n\geq 1$, $\E\left[|\phi_n|^2|\cF_{n-1} \right]\leq
\alpha_n^2 \rho_\cK \int |\phi_{n-1}+y|^2 d\mu_{\theta_{n-1}}(y, w)$.  By
Assumption~\ref{hyp:mutheta}-\ref{hyp:mutheta:item2}), there exists a constant
$C$ such that for any $n \geq 1$ almost-surely
\[
\E\left[|\phi_n|^2|\cF_{n-1} \right] \indic_{\theta_{n-1} \in \cK}\leq \alpha_n^2 \rho_\cK \left(|\phi_{n-1}|^2 +
  2 |\phi_{n-1}| \sqrt{C}+ C \right) \, .
\]
Set $U_n \eqdef|\phi_n|^2\indic_{\bigcap_{j \leq n-1} \{\theta_j \in \cK \}}$.
Upon noting that $\indic_{\bigcap_{j \leq n-1} \{\theta_j \in \cK \}} \leq
\indic_{\bigcap_{j \leq n-2} \{\theta_j \in \cK \}}$, the previous inequality
implies $ \E[U_n]\leq \alpha_n^2 \rho_\cK \left( \E\left[U_{n-1} \right] + 2
  \sqrt{ \E[U_{n-1}]}\, \sqrt{C} + C \right)$.  Let $\delta \in (\rho_\cK, 1)$.
For any $n$ large enough (say $n \geq n_0$), $\alpha_n^2 \rho_\cK \leq 1-
\delta$ since $\lim_n \alpha_n=1$ under
Assumption~\ref{hyp:step}-\ref{hyp:step:item1}). There exist positive constants
$M,b$ such that for any $n \geq n_0$,
\begin{align*}
  \E[U_n] & \leq \left(1- \delta \right) \left( \E\left[U_{n-1} \right] + 2 \sqrt{
      \E[U_{n-1}]}\, \sqrt{C} + C \right)  \\
  & \leq \left(1- \frac{\delta}{2} \right) \E[U_{n-1}] + b\un_{\E[U_{n-1}] \leq
    M}\,.
\end{align*}
A trivial induction implies that $\E[U_n] \leq (1-\delta/2)^{n-n_0} \E[U_{n_0}]
+ 2b /\delta$, which concludes the proof.

\section{Preliminary results on the sequence
  $(\phi_n)_n$}
\label{app:Phi}
Due to the coupling of the sequences $(\athn)_n$ and $(\phi_n)_n$ (see
Eq.~(\ref{eq:SAmarkovnoise})), the asymptotic analysis of $(\athn)_n$ requires
a more detailed understanding of the behavior of $\phi_n$. Note from
Assumption~\ref{hyp:mutheta}-\ref{hyp:mutheta:item1}) and (\ref{eq:thOrth})
that $\{\phi_n, n\geq 0\}$ is a Markov chain w.r.t. the filtration $\{\F_n,
n\geq 0\}$ with a transition kernel controlled by $\{\alpha_n, \theta_{n}, n
\geq 0\}$ (see also (\ref{eq:kernel:phin}) below).

  Let us introduce some notations and definitions. If $(x,A)\mapsto P(x,A)$ is a probability
  transition kernel on $\bR^{dN}$, then for any bounded continuous function
  $f:\bR^{dN}\to \bR$, $Pf$ is the measurable function $x\mapsto \int f(y)
  P(x,dy)$\,.  If $\nu$ is a probability on $\bR^{dN}$, $\nu P$ is the
  probability on $\bR^{dN}$ given by $\nu P(A) = \int \nu(dx) \, P(x,A)$.  For
  $n \geq 0$, notation $P^n$ stands for the $n$-order iterated kernel
  \emph{i.e.}, $P^nf(x) = \int P^{n-1}f(y) P(x,dy)$; by convention $P^0(x,A) =
  \un_A(x) = \delta_x(A)$. A measure $\pi$ is said to be an invariant
  distribution w.r.t. $P$ if $\pi P = \pi$.  For $p\geq 0$, denote by
  $\cL_p(\bR^{dN})$ the set of lipschitz functions $f:\bR^{dN}\to\bR^{dN}$
  satisfying
\[
[f]_p \eqdef  \sup_{x,y \in \bR^{dN}} \frac{|f(x) -f(y)|}{|x-y| (1+|x|^p+ |y|^p)} <
\infty \, .
\]
We define $N_p(f) \eqdef   (\sup_{x \in \bR^{dN}} \frac{|f(x)|}{1+|x|^{p+1}})\vee [f]_p$
for $f \in \cL_p(\bR^{dN})$.  
For any $\bth\in \bR^{dN}$ and any $\alpha\geq 0$, define the probability
transition kernel $P_{\alpha,\bth}$ on $\bR^{dN}$ as
\begin{equation}
  \label{eq:kernel:Palphatheta}
  P_{\alpha,\bth}f(x) = \int f\left(\alpha \cJ_\bot(w\otimes
I_d)(x+y)\right) \ d\mu_\bth(y,w) \, .
\end{equation}
This collection of kernels is related to the sequence $(\phi_n)_n$ since by
Assumption~\ref{hyp:mutheta}-\ref{hyp:mutheta:item1}) and (\ref{eq:thOrth}),
for any measurable positive function $f$ it holds almost-surely
\begin{equation}
  \label{eq:kernel:phin}
  \bE\left[ f(\phi_{n+1})|\cF_n\right] =
P_{\alpha_{n+1},\thn}f(\phi_n)\ .
\end{equation}

We start with a result that claims that any transition kernel
$P_{\alpha,\theta}$ possesses an unique invariant distribution
$\pi_{\alpha,\theta}$ and is ergodic at a geometric rate. This also implies
that for a large family of functions $f$, a solution $f_{\alpha,\theta}$ to
the Poisson equation 
\begin{equation}
  \label{eq:PoissonEquation}
   f - \pi_{\alpha, \theta}(f) = f_{\alpha, \theta} - P_{\alpha,\theta} f_{\alpha,\theta}
\end{equation}
exists, and is unique up to an additive constant.

\begin{prop}
  \label{prop:invariantPI}
  Let Assumptions~\ref{hyp:mutheta} and~\ref{hyp:model} hold.  Let $\cK \subset
  \bR^{dN} $ be a compact set and let $\rho_\cK \in (0,1)$ be given by
  Assumption~\ref{hyp:model}. The following holds for any $a \in
  (0,1/\sqrt{\rho_\cK})$.
\begin{enumerate}
\item \label{prop:invariantPI:item1} For any $\theta \in \cK$ and $\alpha \in
  [0,a]$, $P_{\alpha,\theta}$ admits an unique invariant distribution
  $\pi_{\alpha,\theta}$ such that $
  \sup_{\alpha\in [0,a],\bth\in \cK} \int
  |x|^2d\pi_{\alpha,\bth}(x)<\infty\,.
  $
\item \label{prop:invariantPI:item2} For any $p \in [0,1]$, there exists a
  constant $K$ such that for any $x\in \bR^{dN}$ and any $f \in
  \cL_p(\bR^{dN})$,
$    \sup_{\alpha\in [0,a],\bth\in \cK}
    |P_{\alpha,\bth}^nf(x)-\pi_{\alpha,\bth}(f)| 
    \leq K N_p(f) \, \left(a \sqrt{\rho_\cK} \right)^n \  (1+|x|^{p+1})\, .
  $
\item \label{prop:invariantPI:item3} For any $\alpha \in (0, a]$, $\theta \in
  \cK$, $p \in [0,1]$ and $f \in \cL_p(\bR^{dN})$, the function
  $f_{\alpha,\theta}: x \mapsto \sum_{n \geq 0} \left(P_{\alpha,\theta}^n f(x)
    - \pi_{\alpha,\theta}f\right)$ exists, solves the Poisson equation
  (\ref{eq:PoissonEquation}) and is in $\cL_p(\bR^{dN})$. In addition, 
\[
\sup_{\alpha \in [0,a], \theta \in \cK} |f_{\alpha,\theta}(x)| \leq \frac{K N_p(f) \ }{1-a
  \sqrt{\rho_\cK}} \ (1+|x|^{p+1}) \, .
\]
%where $K$ is given by item~\ref{prop:invariantPI:item2}).
\end{enumerate}
\end{prop}
\begin{proof}
  Let $\cK$ be a compact subset of $\bR^{dN}$. Throughout this proof, for ease
  of notations, we will write $\rho$ instead of $\rho_\cK$. Let $a \in (0,
  1/\sqrt{\rho})$ be fixed. We check the assumptions of \cite[Proposition 2 p.
  253]{benveniste:metivier:priouret:1987} from which all the items follow.  We
  first prove \cite[(2.1.10) p.253]{benveniste:metivier:priouret:1987}. By
  Assumption~\ref{hyp:model}, for any $\alpha \in [0,a]$ and $\theta \in \cK$
  \begin{multline*}
    \int P_{\alpha,\theta}(x,dy) |y|^2  \\
    \leq a^2 \rho \left( |x|^2 + \int |y|^2 d\mu_\theta(y, w) + 2 |x| \int |y|
      d\mu_\theta(y,w) \right) \, ;
  \end{multline*}
  by Assumption~\ref{hyp:mutheta}-\ref{hyp:mutheta:item2}), for any $\bar \rho
  \in (a^2 \rho , 1)$, there exists a positive constant $c$ such that for any
  $x \in \bR^{dN}$
\[
\sup_{\alpha \in [0,a], \theta \in \cK} \int P_{a,\theta}(x,dy) |y|^2 \leq \bar
\rho |x|^2 +c \, .
\]
This concludes the proof of \cite[(2.1.10)
p.253]{benveniste:metivier:priouret:1987}. Note that iterating this inequality
and applying the Jensen's inequality yield for any $n \geq 1$, $p \in [0,1]$,
$x \in \bR^{dN}$,
\begin{equation}
\label{eq:drift}
\sup_{\alpha \in [0,a], \theta \in \cK} \int P^n_{a,\theta}(x,dy) |y|^{p+1} \leq  \left( \bar \rho^n |x|^2 + \frac{c}{1-\bar \rho} \right)^{\frac{p+1}{2}} \,  .
\end{equation} We now prove \cite[(2.1.9)
p.253]{benveniste:metivier:priouret:1987} Let $x,z \in \bR^{dN}$, $\alpha \in
[0, a]$ and $\theta \in \cK$.  We consider a coupling of the distributions
$P_{\alpha,\theta}^n(x,\cdot)$ and $P_{\alpha,\theta}^n(z,\cdot)$ defined as
follows: $(\overline W_n,\overline Y_n)_{n\in \bN}$ are i.i.d. random
variables with distribution $\mu_\bth$ and set $\overline\cW_n = \overline
W_n\otimes I_d$. The stochastic process $(\varphi^x_n)_{n\in \bN}$ defined
recursively by $\varphi_n^x = \alpha \cJo
\overline\cW_n(\varphi_{n-1}^x+\overline Y_n)$ and $\varphi^x_0=x$ is a Markov
chain with transition kernel $P_{\alpha,\theta}$ starting from $x$. We denote
by $\bE_{\alpha,\theta}$ the expectation on the associated canonical space. Let
$p \in [0,1]$. For any $g \in \cL_p(\bR^{dN})$, it holds
  \begin{align}
    & \left| P_{\alpha,\theta}^n g(x) - P_{\alpha,\theta}^n g(z) \right| =
    \left|
      \bE_{\alpha,\theta}\left(g(\phi_n^x) - g(\phi_n^z) \right) \right|  \nonumber \\
    & \leq  \bE_{\alpha,\theta}\left( \left| g(\phi_n^x) - g(\phi_n^z)  \right|\right) \nonumber \\
    & \leq [g]_p \, \bE_{\alpha, \theta} \left[ \left| \phi_n^x- \phi_n^z
      \right| \left(1 +
        |\phi_n^x|^p + |\phi_n^z|^p \right)\right] \nonumber  \\
    & \leq [g]_p \, \left\{ \bE_{\alpha, \theta} \left| \phi_n^x- \phi_n^z
      \right|^2 \, \bE_{\alpha,\theta}\left[ \left( 1 + |\phi_n^x|^p +
          |\phi_n^z|^p \right)^2\right] \right\}^{1/2} \, .
    \label{eq:PoissonContinue}
  \end{align} By
  Assumption~\ref{hyp:model} combined with a trivial induction,
  \begin{align}
    \bE_{\alpha,\theta}(|\varphi^x_{n} &-\varphi^z_{n}|^2)^{1/2} = \alpha
    \bE_{\alpha,\theta}(|\cJo
    \overline\cW_n(\varphi^x_{n-1}-\varphi_{n-1}^z)|^2)^{1/2}  \nonumber \\
    & \hspace{-1cm} = \alpha  \bE_{\alpha,\theta}( (\varphi^x_{n-1}-\varphi_{n-1}^z)^T  \mathsf{A}_\theta (\varphi^x_{n-1}-\varphi_{n-1}^z))^{1/2} \nonumber \\
    & \hspace{-1cm} \leq a \sqrt\rho \,
    \bE_{\alpha,\theta}(|\varphi^x_{n-1}-\varphi_{n-1}^z|^2)^{1/2} \nonumber \\
    & \hspace{-1cm} \leq \left(a \sqrt{\rho} \right)^n |x-z| \, , \label{eq:contraction}
  \end{align}  
  where $\mathsf{A}_\theta \eqdef \int (w \otimes I_d)^T \cJo (w \otimes I_d)
  d\mu_{\theta}(y,w) $.  Combining (\ref{eq:drift}) and (\ref{eq:contraction})
  shows that there exists $C>0$ such that for any $x,z \in \bR^{dN}$, $g \in
  \cL_p(\bR^{dN})$ and $n \geq 1$,
\begin{multline}
  \label{eq:BMP:219}
  \sup_{\alpha \in [0,a], \theta \in \cK} \left| P_{\alpha,\theta}^n g(x) -
  P_{\alpha,\theta}^n g(z) \right|  \\
\leq C\, [g]_p |x-z|\left(a \sqrt{\rho} \right)^n
\left(1+|x|^{p}+ |z|^{p} \right) \, .
\end{multline} 
This concludes the proof of \cite[(2.1.9)
p.253]{benveniste:metivier:priouret:1987}.  Finally, we show that the
transition kernels are weak Feller.  From (\ref{eq:kernel:Palphatheta}) and the
dominated convergence theorem, it is easily checked that for any bounded
continuous function $f$ on $\bR^{dN}$, $x \mapsto P_{\alpha,\theta} f(x)$ is
continuous.  Therefore, all the assumptions of \cite[Proposition 2
p.253]{benveniste:metivier:priouret:1987} are verified.
\end{proof}

% In Proposition~\ref{prop:momentPITheta}, we go further by giving an explicit
% expression of the first two moments of $\pi_{\alpha,\theta}$.
\begin{prop}\label{prop:momentPITheta}
  Let Assumptions~\ref{hyp:mutheta} and~\ref{hyp:model} hold. Let $\theta \in
  \bR^{dN}$ and $\alpha$ such that $\pi_{\alpha,\theta}$ exists.
  \begin{enumerate}
  \item \label{prop:momentPITheta:item1} The first order moment
    $m^{(1)}_\theta(\alpha) \eqdef \int x\,d\pi_{\alpha,\theta}(x)$ of
    $\pi_{\alpha,\theta}$ is given by $m^{(1)}_\theta(\alpha) = (\alpha^{-1} I_{dN}- \cJo\W_\theta)^{-1} \cJo z_\theta$
where $\W_\theta$ and $z_\theta$ are given by (\ref{eq:Omega}) and
(\ref{eq:upsilon}).  
\item \label{prop:momentPITheta:item2} Set $T(w)\eqdef ((\Jo w)\otimes
  I_d)\otimes ((\Jo w)\otimes I_d)$. The vector $m^{(2)}_\theta(\alpha) \eqdef
  \vect(\int xx^Td\pi_{a,\theta}(x))$ is given by $m^{(2)}_\theta(\alpha) = \left(\alpha^{-2}I_{d^2N^2}- \Phi_\theta
  \right)^{-1}\zeta_\theta(\alpha)$
where $  \Phi_\theta \eqdef  \int T(w)d\mu_{\theta}(y,w)$ and
$
  \zeta_\theta(\alpha) \eqdef \int T(w)\vect\left(yy^T+2y\,
    m^{(1)}_\theta(\alpha)^T \right)d\mu_{\theta}(y,w)
$.
  \end{enumerate}
\end{prop}
\begin{proof} Since $\pi_{\alpha,\theta} = \pi_{\alpha,\theta} P_{\alpha,\theta}$,
  we obtain: $   m^{(1)}_\theta(\alpha) =\int\!\!\!\int \alpha \cJ_\bot(w\otimes I_d)(y+x)d\mu_{\theta}(y,w)d\pi_{\alpha,\theta}(x)= \alpha \int ((\Jo w)\otimes I_d)(y+ m^{(1)}_\bth(\alpha))d\mu_{\theta}(y,w)$. This yields the expression of $m^{(1)}_\theta(\alpha)$. The proof of
item~\ref{prop:momentPITheta:item2}) follows the same lines as above and is
omitted. 
\end{proof}
%The following result provides additional regularity-in-$(\alpha,\theta)$ properties. 
The proof of the following Proposition is left to the reader.
\begin{prop}
  \label{prop:Regularity}
  Let Assumptions~\ref{hyp:mutheta},~\ref{hyp:model} and \ref{hyp:lipschitz}
  to hold.  Let $\cK \subset \bR^{dN} $ be a compact set and let $\rho_\cK \in
  (0,1)$ and $\lambda_\mu \in (0,1]$ be given resp. by
  Assumption~\ref{hyp:model} and Assumption~\ref{hyp:lipschitz}.  The following
  holds for any $a \in (0,1/\sqrt{\rho_\cK})$.
\begin{enumerate}
\item \label{prop:Regularity:item1} For any $f \in \cL_1(\bR^{dN})$, there exists a constant $C_f$ such that
  for any $\alpha, \alpha' \in [0,a]$ and $\theta, \theta' \in \cK$,
$  \left|\int f(x) \, \left( d\pi_{\alpha,\theta}(x) - d\pi_{\alpha',\theta'}(x)
    \right) \right| 
  \leq C_f \ \left(|\alpha - \alpha'| + \left|\theta-\theta'
    \right|^{\lambda_\mu} \right) $.
\item \label{prop:Regularity:item2} When $f$ is the identity function $f(x) =x$
  then for any $\alpha \in (0, a]$, $\theta \in \cK$, $x\in \bR^{dN}$, one has
  \begin{equation}
    \label{eq:poisson-identite}
    f_{\alpha,\theta}(x) = (I_{dN}-\alpha\cJo\W_\theta)^{-1}(x-m^{(1)}_\theta(\alpha))\,.
  \end{equation}
  In addition, there exists a constant $K$ such that for any $\alpha,\alpha'
  \in [0, a]$, $\theta, \theta' \in \cK$, one has
$  \left| P_{\alpha, \theta} f_{\alpha,\theta}(x) - P_{\alpha', \theta'}
    f_{\alpha',\theta'}(x) \right| + \left| f_{\alpha,\theta}(x) -
    f_{\alpha',\theta'}(x) \right| 
  \leq K\left( \left|\alpha- \alpha'\right| + \left|\theta -
      \theta'\right|^{\lambda_\mu}\right) \left(1 +|x| \right)$\,.
\item \label{prop:Regularity:item3} For any function $f$ of the form $x\T A x$,
  the Poisson solution $f_{\alpha,\theta}$ exists and there exists a constant
  $K$ such that for any  $\alpha,\alpha'
  \in [0, a]$, $\theta, \theta' \in \cK$, one has
$  \left| P_{\alpha, \theta} f_{\alpha,\theta}(x) - P_{\alpha', \theta'}
    f_{\alpha',\theta'}(x) \right|  \leq K\left(\left|\alpha- \alpha'\right|
    + \left|\theta - \theta'\right|^{\lambda_\mu}\right)  \left(1 +|x|^2
  \right) $.
\end{enumerate}
\end{prop}

\section{Proof of Proposition~\ref{prop:convergence:averagesequence}}
\label{sec:proof:cvg}
\begin{lemma}
\label{lem:deltaTheta}
Under Assumptions~\ref{hyp:step}-\ref{hyp:step:item1}) and ~\ref{hyp:mutheta},
%there exists 
$\exists C>0$ s.t. 
$\left|\theta_{n+1} - \theta_n \right| \leq C\, \gamma_n \left( |Y_{n+1}| +
  |\phi_n |\right)$ a.s.
\end{lemma}
\begin{proof}
  Since $\lim_n \gamma_n / \gamma_{n+1}=1$, there exists a constant $C$ such
  that $   |\theta_{n+1} -\theta_n|\leq |\bs 1 \otimes \la \theta_{n+1}\ra- \bs 1
    \otimes \la \theta_n \ra|+|\cJo\theta_{n+1}|+|\cJo\theta_n| 
     \leq C \left|\la \theta_{n+1} \ra - \la \theta_n \ra \right| + \gamma_{n}
    \phi_{n+1} + \gamma_{n} \phi_n$.
  The result follows from Eqs~(\ref{eq:SAmarkovnoise}), (\ref{eq:thOrth}) and $\sup_n \alpha_n < \infty$.
\end{proof}

\subsection{Decomposition of  $\la \theta_{n+1} \ra-  \la \theta_n \ra$}\label{sec:proof:decomposition}
By (\ref{eq:SAmarkovnoise}), it holds $\la \theta_{n+1} \ra =\la \theta_n \ra
+\gamma_{n+1} h(\athn)+\gamma_{n+1} (\eta_{n+1,1}+\eta_{n+1,2})$ where
$  \eta_{n+1,1} =\la\cW_{n+1}(Y_{n+1}+\phi_{n})\ra\ -
  \la z_{\theta_n} + \W_{\theta_n} \phi_n \ra $, $\eta_{n+1,2} = \la z_{\theta_n} + \W_{\theta_n} \phi_n \ra - h(\athn)$.
We write $\eta_{n+1,2}=u_{n}+v_{n}+w_{n+1}+z_{n}$ where
% \begin{align*}
%   & u_{n} = \la z_{\bth_n} - z_{\cJ \bth_n} \ra \\
%   & v_n =  \la \W_{\bth_n} - \W_{\cJ \bth_n} \ra\ \, \phi_{n} \\
%   & w_{n+1} =  \la \W_{\cJ\thn}\ra (\phi_{n}- m^{(1)}_{\theta_{n}}(\alpha_{n+1}))\, \\
%   & z_n = \la \W_{\cJ\thn}\ra (m^{(1)}_{\theta_{n}}(\alpha_{n+1})) - m^{(1)}_{\cJ\thn}(1))
% \end{align*}
$u_{n} = \la z_{\bth_n} - z_{\cJ \bth_n} \ra$,
$v_n =  \la \W_{\bth_n} - \W_{\cJ \bth_n} \ra\ \, \phi_{n}$,
  $ w_{n+1} =  \la \W_{\cJ\thn}\ra (\phi_{n}- m^{(1)}_{\theta_{n}}(\alpha_{n+1}))$,
  $ z_n = \la \W_{\cJ\thn}\ra (m^{(1)}_{\theta_{n}}(\alpha_{n+1})) - m^{(1)}_{\cJ\thn}(1))$.
We
finally introduce a decomposition of $w_n$. For any compact $\cK$, let
$\rho_\cK \in (0,1)$ be given by Assumption~\ref{hyp:model}. Let $a \in (1, 1
/\sqrt{\rho_\cK})$.  Under Assumption~\ref{hyp:step}, the sequence
$(\alpha_n)_n$ given by (\ref{eq:definition:phin}) converges to one; hence,
there exists a (deterministic) integer $n_0$ (depending on $\cK$) such that $
\alpha_n \in (0, a)$ for all $n\geq n_0$. The identity function is in
$\cL_0(\bR^{dN})$ and by Proposition~\ref{prop:Regularity}, there exists a
solution $gf{\alpha,\bth}$ to the Poisson equation (\ref{eq:PoissonEquation})
with the $f$ equal to the identity function, for any $\alpha \in (0, a)$ and
$\theta \in \cK$; by (\ref{eq:poisson-identite}) $f_{\alpha,\bth}(x) =
(I_{dN}-\alpha\cJo\W_\theta)^{-1}(x-m^{(1)}_{\theta}(\alpha))$.  To make the
notation easier, we will set below $ f_n \eqdef f_{\alpha_{n+1},\thn}$ and $
P_n \eqdef P_{\alpha_{n+1},\thn}$.  By
Proposition~\ref{prop:invariantPI}-\ref{prop:invariantPI:item3}), there exists
a constant $C>0$ such that a.s.
\begin{equation}
\sup_{n \geq n_0} |f_{n}(x)| \indic_{E_\cK} \leq C (1+|x|)\,.
\label{eq:gbornee}
\end{equation}
Letting $x=\phi_n$ in the Poisson equation (\ref{eq:PoissonEquation}), we
obtain $\phi_{n}- m^{(1)}_{\theta_{n}}(\alpha_{n+1}) = f_n(\phi_n) - P_n f_n(\phi_n)$.
We set $w_{n+1}= e_{n+1}+c_{n+1}+s_{n+1}+t_{n}$ where
$ e_{n+1} = \la \W_{\cJ\thn}\ra \left( f_{n}(\phi_{n+1})- P_{n} f_{n}(\phi_{n})
  \right)$,
  $ c_{n+1} =  \la \W_{\cJ\thn}\ra  \, f_{n-1}(\phi_{n}) - \la \W_{\cJ\theta_{n+1}}\ra \, f_{n}(\phi_{n+1})$,
  $ s_{n+1} = \la \W_{\cJ\theta_{n+1}} - \W_{\cJ\thn}\ra  \, f_{n}(\phi_{n+1})$ and
  finally $ t_{n} = \la \W_{\cJ\thn}\ra \, \left( f_{n}(\phi_{n}) - f_{n-1}(\phi_{n}) \right)$.
As a conclusion, we have $\eta_{n+1,2} = u_n + v_n + z_n +  e_{n+1} + c_{n+1} +s_{n+1} +t_n$.

\subsection{Proof of Proposition~\ref{prop:convergence:averagesequence}}

Define $E_\cK=\{ \forall j\in \bN,\,\theta_j\in \cK\}$ and
$E_{n,\cK}=\cap_{j\leq n}\{\theta_j\in\cK\}$ for some compact set $\cK$.

We show that $\sum_n\gamma_n\eta_{n,i}<\infty$ a.s.  for both $i=1,2$. The
proposition will then follow from \cite{andrieu:2005}. By
Assumption~\ref{hyp:stability}, it is enough to show that for any fixed compact
set $\cK$, $\sum_{k \geq 1}\gamma_k \eta_{k,i}\indic_{E_\cK}$ is finite a.s.
Hereafter, $\cK$ is fixed and $n_0$ is defined as in
Section~\ref{sec:proof:decomposition}.

We first study $\eta_{n,1}$. Note that for any $\omega$, the sequence
$\indic_{E_{n,\cK}}(\omega)$ is identically equal to $\indic_{E_{\cK}}(\omega)$
for all large $n$. As a consequence, $\sum_n\gamma_n
\eta_{n,1}(\indic_{E_{\cK}}-\indic_{E_{n-1,\cK}})$ is finite a.s. and it is
therefore sufficient to prove that $\sum_n\gamma_n
\eta_{n,1}\indic_{E_{n-1,\cK}}$ is finite a.s. Since
$\eta_{n,1}\indic_{E_{n-1,\cK}}$ is a martingale difference noise, the sought
result will be obtained provided $\sum_n\gamma_n^{1+\lambda}
\bE[|\eta_{n,1}|^{1+\lambda}\indic_{E_{n-1,\cK}}]<\infty$ where $\lambda >0$
(see \emph{e.g.} \cite[Theorem 2.18]{hall:heyde:1980}); we choose $\lambda \in (0,1)$
given by Assumption~\ref{hyp:step}.  After some algebra, $\sup_n
\bE[|\eta_{n,1}|^2\indic_{E_{n-1,\cK}}] \leq 2
\sup_n\bE[|\la\cW_n(Y_n+\phi_{n-1})\ra|^2\indic_{E_{n-1,\cK}}]\leq C \sup_n
\bE[(|Y_n|^2+|\phi_{n-1}|^2)\indic_{E_{n-1,\cK}}]$ for some constant $C$ -
where we used the fact that $W_n$ is row-stocahstic and thus has bounded entries.
Assumption~\ref{hyp:mutheta}-\ref{hyp:mutheta:item2}) directly leads to
$\sup_n\bE[|Y_n|^2\indic_{E_{n-1,\cK}}]<\infty$ whereas by
Lemma~\ref{lem:agreement},
$\sup_n\bE[|\phi_{n-1}|^2\indic_{E_{n-1,\cK}}]<\infty$.  Hence,
$\sum_n\gamma_n^{1+\lambda} \bE[|\eta_{n,1}|^{1+\lambda}\indic_{E_{n-1,\cK}}]
\leq C' \sum_n \gamma_n^{1+\lambda}$ for some $C'>0$. And the upper bound is
finite by Assumption~\ref{hyp:step}. This concludes the first step.

We now study $\eta_{n,2}$ for any $n \geq n_0$.  By (\ref{eq:lipschitz1}),
there exists $C$ such that $|u_n|\indic_{E_\cK}\leq
C|\cJo\thnmu|^{\lambda_\mu}\indic_{E_\cK}\leq C \gamma_n^{\lambda_\mu} |\phi_{n-1}|^{\lambda_\mu}\indic_{E_{n-2,\cK}}$\,.
Therefore, $\bE(\indic_{E_\cK}\sum_{n} \gamma_n |u_n|)\leq C \sum_{n}
\gamma_n^{1+\lambda_\mu} \sup_n \bE(|\phi_{n-1}|\indic_{E_{n-2,\cK}})$ which is finite by
Assumption~\ref{hyp:step} and Lemma~\ref{lem:agreement}. Thus $\sum_{n}
\gamma_n |u_n| \indic_{E_\cK}$ is a.s. finite.

The term $v_n$ can be analyzed similarly: by (\ref{eq:lipschitz2})
applied with $\cK \leftarrow \cK \cup \{\cJ \theta, \theta \in \cK \}$, there
exists a constant $C$ such that $|v_n|\indic_{E_\cK}\leq C |\cJo\thn|^{\lambda_\mu}
\,|\phi_{n}|\indic_{E_{n-1,\cK}} \leq C \gamma_{n+1}^{\lambda_\mu}
|\phi_{n}|^{1+\lambda_\mu} \indic_{E_{n-1,\cK}}$ and the fact that $\sum_n\gamma_n|v_n|
\indic_{E_\cK}$ is finite a.s. follows from the same arguments as above.

We now study  $|z_n|\leq
C_v\,|m^{(1)}_{\theta_{n}}(\alpha_{n+1}) - m^{(1)}_{\cJ\thn}(1)|$.  By
Proposition~\ref{prop:Regularity}-\ref{prop:Regularity:item1}), since $\alpha_{n+1} < a
<1/\sqrt{\rho_\cK}$, there exists a constant $C'$ such that
$  \sum_n \gamma_n \bE(|z_n| \indic_{E_\cK}) $ is no larger than $C' \sum_n | \gamma_n - \gamma_{n+1} | +
    \gamma_n^{1+\lambda_\mu} \sup_k
    \bE(|\phi_{k}|^{\lambda_\mu}\indic_{E_{k-1,\cK}})  $.
The latter is finite by Lemma~\ref{lem:agreement} and Assumption~\ref{hyp:step}.
Hence, $\sum_n \gamma_n |z_n| \indic_{E_\cK}$ is finite a.s.

$(e_n)_n$ is a martingale-increment sequence: as above for the term
$\eta_{n,1}$, $\sum_n \gamma_n e_n \indic_{E_\cK}$ is finite a.s. if $\sup_n
\bE(|e_{n+1}|^{1+\lambda} \indic_{E_{n,\cK}}) < \infty$. This holds true by (\ref{eq:gbornee}) and Lemma~\ref{lem:agreement}.

Let us now investigate $c_{n+1}$. We write
$  \sum_{k=1}^n\gamma_{k+1} c_{k+1} = \sum_{k=2}^{n}(\gamma_{k+1}-\gamma_k) \la \W_{\cJ\theta_k}\ra  f_{k-1}(\phi_{k}) 
  -\gamma_{n+1} \la \W_{\cJ\theta_{n+1}}\ra f_n(\phi_{n+1})+\gamma_2 \la
  \W_{\cJ\theta_1}\ra f_0(\phi_1)\,.$
Using again (\ref{eq:gbornee}) and
Lemma~\ref{lem:agreement}, there exists $C>0$ such that
$
\sum_{k=1}^n\gamma_{k+1} \bE\left( |c_{k+1}| \indic_{E_{\cK}} \right) \leq C \left(\sum_{k\geq1} |
  \gamma_{k+1}-\gamma_k| + \gamma_n + 1 \right) \, .
$
The right hand side is finite by Assumption~\ref{hyp:step}, thus implying that $\sum_n
\gamma_n c_n \indic_{E_\cK}$ is finite a.s.

Consider the term $s_{n+1}$. Following similar arguments and using
(\ref{eq:gbornee}) again, we obtain
$$
\sum_{k\leq n}\gamma_k |s_k|\indic_{E_{\cK}}\leq C\sum_{k\leq n}\gamma_k\|\la \W_{\cJ\theta_k} - \W_{\cJ \theta_{k-1}}\ra \| (1+|\phi_k|)\indic_{E_{\cK}}
$$
for some constant $C$ which depends only on $\cK$. By
condition~(\ref{eq:lipschitz2}) and Lemma~\ref{lem:deltaTheta}, one has $\|\la
\W_{\cJ\theta_k} - \W_{\cJ \theta_{k-1}}\| \indic_{E_\cK}\leq C_\cK \gamma_{k}^{\lambda_\mu}
\left( |Y_{k}|^{\lambda_\mu} + | \phi_{k-1}|^{\lambda_\mu}\right)\indic_{E_{\cK}}$.  By Cauchy-Schwarz
inequality, Assumption~\ref{hyp:mutheta} and Lemma~\ref{lem:agreement}, it can
be proved that
\begin{equation}
  \label{eq:controle:moments:Yphi}
   \sup_k \bE\left[(|Y_k|+|\phi_{k-1}|) (1+|\phi_k|)\indic_{E_{\cK}}
\right] < \infty \, .
\end{equation}
By Assumption~\ref{hyp:step}, $\E(\sum_{k}\gamma_k
|s_k|\indic_{E_{\cK}})$ is finite thus implying that $ \sum_{k\geq 1}
\gamma_ks_k\indic_{E_\cK}$ exists a.s.

Finally consider the term $t_{n}$.  By 
Proposition~\ref{prop:Regularity}-\ref{prop:Regularity:item2}), there exists a
constant $C$ such that for any $n \geq n_0$,$ |t_{n}| \indic_{E_\cK} \leq C
\left(|\alpha_{n} - \alpha_{n-1}| + |\theta_{n} - \theta_{n-1}|^{\lambda_\mu}
\right) \left(1 + |\phi_{n}| \right)$.  By Lemma~\ref{lem:deltaTheta},
(\ref{eq:controle:moments:Yphi}) and Assumption~\ref{hyp:step}, it can be shown
that $\sum_n \gamma_n \bE(|t_n| \indic_{E_\cK}) < \infty$ which proves that
$\sum_n \gamma_n t_n \indic_{E_\cK}$ converges a.s.

%%%%%%%%%%%%%%%%%%%%%%%%%%%%%%%%%%%%%%%%%%%%%%%%%%%%%%%%%%%%%%%%%%%%%%%%%%
\section{Proof of Theorem~\ref{theo:TCL}}
\label{sec:proof:TCL}
The core of the proof consists in checking the conditions of \cite[Theorem
2.1]{fort:2013}.  To make the notations easier, we write the proofs in the case
$d=1$ and under the assumption that $\lim_n \theta_n = \theta_\star \un$
almost-surely.  
Throughout the proof, we will write that a sequence of r.v. $(Z_n)_n$ is
$O_{w.p.1}(1)$ iff $\sup_n |Z_n| < \infty$ almost-surely; and $(Z_n)_n$ is
$O_{L^1}(1)$ iff $\sup_n \EE{|Z_n|} < \infty$.

%% Observe that under Assumption~\ref{hyp:TCL:pas}, we have the following
%% properties which will be used throughout the proof
%%   \begin{enumerate}
%%   \item $\lim_n |\alpha_n -1 | /\sqrt{\gamma_n} = 0$. \textcolor{cyan}{ok on
%%       s'en sert}
%%   \item $\lim_n |\alpha_{n+1} - \alpha_n| /\sqrt{\gamma_n} = 0$. \textcolor{cyan}{ok on
%%       s'en sert}
%%   \item $\sum_n|\alpha_{n} -1|^{1+\tau/2}< \infty$. \textcolor{cyan}{ok on
%%       s'en sert}
%%   \item $\sum_n \gamma_n^{1+\lambda_\mu} < \infty$ \textcolor{cyan}{ok on
%%       s'en sert}
%%   \item $\lim_n \gamma_n n^{1/(1+\tau/2)} =0$. \textcolor{cyan}{ok on
%%       s'en sert}
%%   \end{enumerate}

Fix $\delta >0$. Set for any positive integers $m \leq k$
$\cA_m \eqdef \bigcap_{j \geq m} \{ |\theta_j - \theta_\star \un| \leq \delta$.
From Section~\ref{sec:proof:decomposition}, it holds
$
\la \theta_{n+1} \ra= \la \theta_n \ra + \gamma_{n+1} h(\la \theta_n \ra) +
\gamma_{n+1} E_{n+1} + \gamma_{n+1} R_{n+1}
$
where $  E_{n+1}  \eqdef \la \cW_{n+1} \left(Y_{n+1} + \phi_n \right) \ra - \left( \la
  z_{\theta_n}  \ra + \la \W_{\theta_n} \ra \phi_n  \right) 
  + \la \W_{\cJ \theta_n} \ra  \left( f_n(\phi_{n+1}) - P_n f_n(\phi_n) \right)$ and where
  $R_{n+1}  \eqdef  u_n+ v_n+ z_n +c_{n+1} +s_{n+1} +t_n$.
Note that $\EE{ E_{n+1} \vert \F_n} =0$ \emph{i.e.}, $(E_n)_n$ is a $\cF_n$-adapted
martingale increment. From the expression of $f_n = f_{\alpha_{n+1}, \theta_n}$
(see Proposition~(\ref{eq:poisson-identite})), we have
\begin{equation}
  \label{eq:Proof:TCL:poisson}
  f_{\alpha, \theta}(y) - { P_{\alpha,\theta} f_{\alpha, \theta}(x)} =
\B_{\alpha,\theta} \left(y - \alpha \cJ_\bot \W_\theta x - \alpha
 \cJ_\bot  z_\theta\right)
\end{equation}
with $\B_{\alpha,\theta} \eqdef \left(I_{dN} - \alpha \cJ_\bot \W_\theta
\right)^{-1}$. Hence, 
\small
\begin{align*}
E_{n+1}  = &\, \la \cW_{n+1} \left(Y_{n+1} + \phi_n \right) \ra - \la
  z_{\theta_n}  \ra - \la \W_{\theta_n} \ra \phi_n \\
  &\, \,  + \la \W_{\cJ \theta_n} \ra \B_{\alpha_{n+1},\theta_n} \left( \phi_{n+1} -
    \alpha_{n+1} \cJ_\bot \left( \W_{\theta_n} \phi_n + z_{\theta_n}
    \right)\right).
    \end{align*}
\normalsize

\vspace{-0.3cm}

\subsection{Checking condition C2 of \cite[Theorem
  2.1]{fort:2013}} 

   We start with a preliminary Lemma which extends Lemma~\ref{lem:agreement}. The proof follows the same line and is thus omitted.
\begin{lemma}\label{lem:TCL:agreementTau} 
  Let Assumptions~\ref{hyp:step}-\ref{hyp:step:item1}),~\ref{hyp:mutheta},
  \ref{hyp:TCL:momentY} and \ref{hyp:TCL:model} hold.  Let $(\phi_n)_{n \geq
    0}$ be the sequence given by (\ref{eq:definition:phin}) and $\tau$ be given
  by Assumption~\ref{hyp:TCL:momentY}.  For any compact set $\cK\subset
  \bR^{dN}$,
\[
\sup_n \bE \left(|\phi_n|^{2+\tau} \indic_{\bigcap_{j \leq n-1} \{\theta_j \in
    \cK \}} \right)<\infty \, .
\]
Let $\tilde \rho_\cK$ be given by Assumption~\ref{hyp:TCL:model}. For any $a
\in (0,1 /\sqrt{\tilde \rho_\cK})$,
$
\sup_{\alpha \in [0,a], \theta \in \cK} \int |x|^{2+\tau}
d\pi_{\alpha,\theta}(x) < \infty $.
\end{lemma}
Let $m \geq 1$.  From Assumption~\ref{hyp:TCL:momentY}
and Lemma~\ref{lem:TCL:agreementTau}, it is easily seen from the above expression of $E_{n+1}$ 
% by using the expression (\ref{eq:expression:En}) 
that $\sup_n \EE{\left|
    E_{n+1}\right|^{2 +\tau} \un_{\bigcap_{m \leq j \leq n} \{ |\theta_j - \theta_\star \un| \leq \delta
\} }} < \infty$ where $\tau$ is given
by Assumption~\ref{hyp:TCL:momentY}.

In order to derive the asymptotic covariance, we go further in the
expression of the conditional covariance $\EE{E_{n+1}^2 \vert \cF_n }$.
We write $\EE{E_{n+1}^2 \vert \cF_n }= \Xi(\alpha_{n+1},
\theta_n,\phi_n)$ where $\Xi(\alpha, \theta,x) \eqdef \int \left( \xi_{\alpha, \theta,x}(y,w) \right)^2 \ d\mu_{\theta}(y, w)$
\begin{equation}
\xi_{\alpha, \theta,x}(y,w) \eqdef \A_{\alpha,\theta} \left(\left(w  - \W_\theta \right) x + \left(  w y  - z_\theta \right) \right)  \label{eq:Expression:xi} 
\end{equation}
and $\A_{\alpha,\theta} \eqdef \frac{\un\T}{N} \left( I_{dN}+ \alpha \, \W_{\cJ \theta} \,  \left(I_{dN} - \alpha \cJ_\bot \W_\theta
\right)^{-1}  \cJ_\bot  \right).$
Set $\pi_\star \eqdef \pi_{1, \theta_\star \un}$ and $\pi_n \eqdef \pi_{\alpha_{n+1},
  \theta_n}$ where $\pi_{\alpha,\theta}$ is defined by
Proposition~\ref{prop:invariantPI}.  We write
\small
\begin{align*}
  & \Xi(\alpha_{n+1}, \theta_n,\phi_n) = \Xi(\alpha_{n+1}, \theta_n,\phi_n) -
  \Xi(1, \theta_n, \phi_n)   \\
  & + \int \Xi(1, \theta_n, x) d \pi_n(x) - \int \Xi(1, \theta_\star \un, x) d
  \pi_{\star}(x)  \\
  & + \Xi(1,\theta_n, \phi_n) - \int \Xi(1,\theta_n, x) d \pi_n(x) \\
  & + \int \Xi(1, \theta_\star \un, x) d \pi_{\star}(x) \ .
\end{align*}
\normalsize

For any $m \geq 1$, we have on the set $\cA_m$

%\small
\begin{align*}
  & \left( \Xi(\alpha_{n+1}, \theta_n,\phi_n) - \Xi(1, \theta_n, \phi_n) \right) \to  0 \ \text{a.s.} \\
  & \left( \int \Xi(1, \theta_n, x) d \pi_n(x) - \int \Xi(1,
    \theta_\star \un, x) d \pi_{\star}(x) \right) \to  0 \ \text{a.s.} \\
  &  \gamma_n \E{\left| \sum_{k=1}^n \left\{\Xi(1,\theta_k, \phi_k) -
        \int \Xi(1,\theta_k, x) d \pi_l(x) \right\} \right| \un_{\cA_m}} \to 0 \ .
\end{align*}
%\normalsize

The detailed computations are given in Section~\ref{sec:longlongdemo}. This
implies that the key quantity involved in the asymptotic covariance matrix is $
\int \Xi(1, \theta_\star \un, x) d \pi_{\star}(x) $.

\vspace{-0.3cm}

\subsection{Expression of $U_\star$} 
Set
$
U_\star \eqdef  \int \Xi(1, \un\otimes\theta_\star,x) \, d\pi_{1,\un\otimes\theta_\star}(x)\,.
$
Lemma~\ref{lem:Ustar} gives an explicit expression for $U_\star$.

\begin{lemma}
\label{lem:Ustar}
Under the assumptions of Theorem~\ref{theo:TCL},
$
\vect U_\star = (\A_\star\otimes \A_\star) (\cR_\star\,m^{(2)}_\star +2\cT_\star
m^{(1)}_\star +\cS_\star).
$
\end{lemma}
\begin{proof}
    For simplicity, we use the notations $R_\theta(w)\eqdef w-\W_\theta$ and
    $\upsilon_\theta(y,w) \eqdef w y- z_\theta$ and
$
\tilde T_{\theta,x}(y,w) \eqdef
(R_\theta(w)x+\upsilon_\theta(y,w))(R_\theta(w)x+\upsilon_\theta(y,w))^T \ .
$
Note that $\tilde T_{\theta,x}(y,w)$ coincides with
$
R_\theta(w)xx^TR_\theta(w)^T +2 R_\theta(w)x\upsilon_\theta(y,w)^T+ \upsilon_\theta(y,w)\upsilon_\theta(y,w)^T\,.
$
From (\ref{eq:Expression:xi}), $\xi_{\alpha, \theta,x}(y,w) =
\A_{\alpha,\theta} (R_\theta(w)x+\upsilon_\theta(y,w))$ so that
$
\vect \Xi(\alpha, \theta,x) = (\A_{\alpha,\theta}\otimes \A_{\alpha,\theta}) \int  \vect \tilde T_{\theta,x}(y,w)\ d\mu_{\theta}(y, w) \ .
$
Applying the $\vect$ operator on $\tilde T_{\theta,x}(y,w)$ yields
$  (R_\theta(w)\otimes R_\theta(w))\vect(xx^T) + 2 (\upsilon_\theta(y,w)\otimes R_\theta(w)) x 
  + \vect (\upsilon_\theta(y,w)\upsilon_\theta(y,w)^T)\,.$
When applied with $\alpha=1$ and $\theta = \theta_\star \un$, it holds
$
\vect \Xi(1, \theta_\star \un,x) = (\A_\star \otimes
\A_\star) (\cR_\star\,\vect(xx^T) +2\cT_\star x +\cS_\star)\,.
$
This yields the result by integrating $x$ w.r.t.  $\pi_\star$.
\end{proof}

\vspace{-0.2cm}

\subsection{Checking condition C3 of \cite[Theorem
2.1]{fort:2013}}
We first prove that for any $m \geq 1$,
\begin{equation}
  \label{eq:proofTCL:step1:reste}
  \left|  u_n  + v_n + z_n +   s_{n+1} + t_n \right|  \un_{\cA_m} \leq \sqrt{\gamma_n}
o(1) O_{L^1}(1) \ .
\end{equation} Let $m \geq 1$. By (\ref{eq:definition:phin}) and
Proposition~\ref{prop:Regularity}-\ref{prop:Regularity:item1}), there exists
a constant $C_1$ such that almost-surely on the set $\cA_m$,
$  \left|z_n \right|  \leq C_1 \left( |\alpha_{n+1} -1| + |\Jo
    \theta_n |^{\lambda_\mu} \right)  \leq C_1 \left( |\alpha_{n+1} -1| + \gamma_{n+1}^{\lambda_\mu} \right)
  \left(1+| \phi_n|^{\lambda_\mu} \right)$.
Assumption~\ref{hyp:TCL:pas}, Lemma~\ref{lem:agreement} and $\lambda_\mu >
1/2$ imply that $ \left| z_n \right|\un_{\cA_m} = \sqrt{\gamma_n}
o(1)O_{L¹}(1)$. By Assumption~\ref{hyp:lipschitz},
Proposition~\ref{prop:invariantPI}-\ref{prop:invariantPI:item3}) and
Lemma~\ref{lem:deltaTheta}, there exist a constant $C_2>0$ and $n \geq n_0$
such that almost-surely, for all $n \geq n_0$,
$
\left|s_{n+1} \right| \un_{\cA_m} \leq C_2 \gamma_n^{\lambda_\mu}
\left(|Y_{n+1}|^{\lambda_\mu} + |\phi_n|^{\lambda_\mu} \right) \left(1 +
  |\phi_{n+1}| \right)\un_{\cA_m} \,.
$ 
Assumption~\ref{hyp:mutheta}, Lemma~\ref{lem:agreement} and the condition
$\lambda_\mu > 1/2$ imply that $\left|s_{n+1} \right| \un_{\cA_m} =
\sqrt{\gamma_n} O_{L^1}(1)$.  By Proposition~\ref{prop:Regularity}-\ref{prop:Regularity:item2}) and
Lemma~\ref{lem:deltaTheta}, there exist a constant $C_3>0$ and $n_0$ such that
almost-surely, for any $n \geq n_0$,
$
\left|t_{n} \right| \un_{\cA_m} \leq C_3 \left( |\alpha_{n+1} - \alpha_{n}| +
  \gamma_n^{\lambda_\mu} \left( |Y_{n}|^{\lambda_\mu} + |\phi_n|^{\lambda_\mu}
  \right)\right)\un_{\cA_m}$.
Lemma~\ref{lem:agreement}, Assumption~\ref{hyp:TCL:pas} and $\lambda_\mu>1/2$
imply that $\left|t_{n+1} \right| \un_{\cA_m} = \sqrt{\gamma_n} o(1)
O_{L^1}(1)$.  By Assumption~\ref{hyp:lipschitz}, there exists a constant
$C_4>0$ such that almost-surely, $ \left| u_n \right| \un_{\cA_m} \leq C_4
\gamma_n^{\lambda_\mu} |\phi_n|^{\lambda_\mu}\un_{\cA_m}$.
Lemma~\ref{lem:agreement} and the property $\lambda_\mu >1/2$ imply $u_n =
o(\sqrt{\gamma_n}) O_{L^1}(1)$. Finally, by Assumption~\ref{hyp:lipschitz},
there exists a constant $C$ such that almost-surely, $\left| v_n \right|
\un_{\cA_m} \leq C \gamma_{n+1}^{\lambda_\mu} |\phi_n|^{1+\lambda_\mu}
\un_{\cA_m}$ so that by Lemma~\ref{lem:agreement} again and the condition
$\lambda_\mu >1/2$, $v_n = o(\sqrt{\gamma_n}) O_{L^1}(1)$.
The above discussion concludes the proof of (\ref{eq:proofTCL:step1:reste}).

The second step is to prove that for any $m \geq 1$, $\sqrt{\gamma_n}
\sum_{k=1}^n c_k \un_{\cA_m} = o(1) O_{w.p.1.}(1) O_{L^1}(1)$. By (\ref{eq:gbornee}), there exists a constant $C>0$ such
that almost-surely,
\[
\left| \sum_{k=1}^n c_k \right| \un_{\cA_m} \leq C \left( 1 + |\phi_0| +
  |\phi_n| \right) \un_{\cA_m}\ .
\]
Lemma~\ref{lem:agreement} implies that $\sum_{k=1}^n c_k = O_{L^1}(1)$.
This concludes the proof of the condition C3 in \cite{fort:2013}.

\vspace{-0.3cm}

\subsection{Detailed computations for verifying the condition C2}
\label{sec:longlongdemo}
The proof of the following lemma follows from
standard computations and is thus omitted.
\begin{lemma} 
\label{lem:TCL:C2} Let Assumptions~\ref{hyp:mutheta}, \ref{hyp:TCL:model} and \ref{hyp:TCL:ContinuiteMuTheta}-\ref{hyp:TCL:ContinuiteMuTheta:item2}) to hold.
Let $\delta >0$ and set $\cK \eqdef \{ \theta: |\theta - \theta_\star \un |
\leq \delta \}$.  Fix $a \in (0, 1/\sqrt{\tilde \rho_\cK})$ where $\tilde
\rho_\cK$ be given by Assumption~\ref{hyp:TCL:model}.  There exists a constant
$C$ such that for any $\theta, \theta'\in \cK$, $\alpha,\alpha' \in [0,a]$,
$x, z, y \in \Rset^{dN}$ and $w \in \mathsf M_1$
\begin{align*}
  & \left| \xi_{\alpha,\theta,x}(y,w)\right| \leq C \left(1 + |y| + |x| \right) \, , \\
  & \| \A_{\alpha,\theta} - \A_{\alpha',\theta'}\| \leq C \left( \left|\alpha-
      \alpha'
    \right| + \left| \theta- \theta' \right|^{\lambda_\mu} \right) \, , \\
  & \left|\xi_{\alpha, \theta,x}(y,w) - \xi_{\alpha',\theta',x}(y,w) \right|
  \\
  & \quad \leq C \left( \left|\alpha- \alpha' \right| + \left| \theta- \theta'
    \right|^{\lambda_\mu} \right) \left(1+ |x| + |y| \right) \ , \\
  & \left|\xi_{\alpha, \theta,x}(y,w) - \xi_{\alpha,\theta,z}(y,w) \right| \leq C \ 
  \left|x-z\right| 
\end{align*}
where $\lambda_\mu$ is given by Assumptions~\ref{hyp:mutheta} and
\ref{hyp:TCL:ContinuiteMuTheta}-\ref{hyp:TCL:ContinuiteMuTheta:item2}).
\end{lemma}

\subsubsection{First term: $\Xi(\alpha_{n+1}, \theta_n,\phi_n)  - \Xi(1, \theta_n,\phi_n)$}
It is sufficient to prove that this term converges almost-surely to zero along
the event $\cA_m$, for any $m \geq 1$; which is implied by the almost-sure
convergence to zero along the event $\theta \in \cK \eqdef \{\theta: |\theta -
\theta_\star| \leq \delta \}$. Below, $C_m$ is a constant whose value may change
upon each appearance. By using the inequality $|a^2 -b^2| \leq |a-b| (|a| +
|b|)$, Assumption~\ref{hyp:TCL:momentY} and Lemma~\ref{lem:TCL:C2}, there
exists a constant $C_m$ such that for any $\alpha$ close enough to~1 and
$\theta \in \cK$, $\left| \Xi(\alpha, \theta,x) - \Xi(1,\theta,x) \right| \leq C_m
  \left(1+|x|^2 \right) \, \left|\alpha- 1 \right|$.
By Lemma~\ref{lem:TCL:agreementTau}, for any $\varepsilon >0$, there exists
$C_m$ such that
$\PP{\sup_{n \geq \ell} (1+|\phi_n|)^2 |\alpha_{n+1} -1 |\un_{\theta_n
    \in \cK} \geq \varepsilon}$ is no larger than
$C_m \sum_{n \geq \ell} \left|\alpha_{n+1}-1 \right|^{(1+\tau/2)}$.
The latter term converges to zero as $\ell \to \infty$ by Assumption~\ref{hyp:TCL:pas}.
This implies that almost-surely, $\lim_n \left| \Xi(\alpha_{n+1},
  \theta_n,\phi_n) - \Xi(1, \theta_n,\phi_n) \right| \un_{\theta_n \in \cK}
=0$.

\subsubsection{Second term: $\int \Xi(1, \theta_n, x)  d\pi_n(x)  - \int \Xi(1, \theta_\star \un,x) d\pi_\star(x)$}
We apply the following lemma (see \cite[Proposition
4.3.]{fort:moulines:priouret:2012}).
\begin{lemma}
  \label{lem:fmp12}
  Let $\mu, \{\mu_n, n \geq 0 \}$ be probability distributions on $\bR^{d N}$
  endowed with its Borel $\sigma$-field. Let $\{h_n, n\geq 1 \}$ be an
  equicontinuous family of functions from $\bR^{dN}$ to $\bR$. Assume
  \begin{enumerate}
  \item the sequence $\{\mu_n, n\geq 0\}$ weakly converges to $\mu$.
  \item for any $x \in \bR^{dN}$, $\lim_n h_n(x)$ exists, and there exists $a
    >1$ such that $\sup_n \int |h_n|^a d \mu_n + \int |\lim_n h_n| d \mu < \infty$.
  \end{enumerate}
  Then $\lim_n \int h_n d \mu_n = \int \lim_n h_n \ d \mu$.
\end{lemma}
\paragraph{Almost-sure weak convergence}
In our case $\mu_n \leftarrow \pi_n$ and $\mu \leftarrow \pi_\star$ and $\mu_n$
is a random probability.  Since the set of bounded Lipschitz functions is
convergence determining (see \emph{e.g.} \cite[Theorem 11.3.3.]{dudley:2002}), we
prove that for any bounded and Lipschitz function $h$, $\lim_n \int h d \pi_n =
\int h d\pi_\star$ almost-surely, with an almost-sure set which has to be
uniform for the set of bounded Lipschitz functions. Following the same lines as
in the proof of \cite[Proposition 5.2.]{fort:moulines:priouret:2012}, this
convergence occurs almost-surely if and only if for any bounded Lipschitz function $h$,
there exists a full set such that on this set, $\lim_n \int h d\pi_n = \int h
d\pi_\star$.

Let $h$ be a bounded Lipschitz function. Then $h \in \cL_0(\bR^{dN})$.  By
Proposition~\ref{prop:Regularity}-\ref{prop:Regularity:item1}), there exists a
constant $C_f$ such that for any $n$ large enough, on the set $\{\theta_n \in
\cK \}$
$
\left| \int h d \pi_n - \int h d\pi_\star \right| \leq C_f \ 
\left(\left|\alpha_{n+1} -1 \right| + \left| \theta_{n+1} - \theta_\star
    \un\right|^{\lambda_\mu} \right)$.
Since $\lim_n \theta_n = \theta_\star \un$ almost-surely and $\lim_n \alpha_n
=1$, we have $\lim_n \int h d\pi_n = \int h d\pi_\star$ almost-surely. This
concludes the proof of the a.s. weak convergence.

\paragraph{Equicontinuity of the family of functions}
\label{sec:Xi:Equicontinuity}
We prove that the family of functions $\{ x \mapsto \Xi(1, \theta,x); \theta
\in \cK \}$ is equicontinuous. Using again the inequality $|a^2 - b^2| \leq
|a-b| (|a| +|b|)$, Lemma~\ref{lem:TCL:C2} and Assumption~\ref{hyp:TCL:momentY},
we know there exists a constant $C_m$ such that for any $\theta \in \cK$, $x,z
\in \Rset^{dN}$, $\left|\Xi(1,\theta,x) - \Xi(1,\theta,z)\right| \leq C_m \, (1+|x| +|z|) |x-z|$.

\paragraph{Almost-sure limit of $ \Xi(1, \theta_n,x)$ when $n \to \infty$} \label{sec:Xi:RegularityInTheta} Let $x$ be fixed.
We write
\small

\vspace{-0.2cm}

\begin{align*}
  & \left| \Xi(1,\theta,x) - \Xi(1,\theta', x) \right| \\
  & \leq \int \left| \xi^2_{1,\theta,x}(y,w) - \xi^2_{1,\theta',
      x}(y,w) \right| d \mu_{\theta'}(y,w) \\
  &+ \left|\int \xi^2_{1,\theta,x}(y,w) d\mu_{\theta}(y,w) - \int
  \xi^2_{1,\theta,x}(y,w) d\mu_{\theta'}(y,w) \right| \ .
\end{align*}

\normalsize

Let us consider the first term. Using again $|a² - b^2| \leq |a-b|(|a| + |b|)$
and Lemma~\ref{lem:TCL:C2}, there exists a constant $C_m$ such that the first
term is upper bounded by $C_m \, (1+|x|^2) |\theta - \theta_\star
\un|^{\lambda_\mu}$ for any $\theta \in \cK$. For the second term, we use
Assumption~\ref{hyp:TCL:ContinuiteMuTheta}-\ref{hyp:TCL:ContinuiteMuTheta:item2bis})
and obtain the same upper bound. Then, there exists a constant $C_m$ such that
for any $\theta, \theta' \in \cK$
\begin{equation}
  \label{eq:TCL:proof:Xi:RegulTheta}
  \left| \Xi(1,\theta,x) - \Xi(1,\theta', x) \right| \leq C_m \ 
(1+|x|^2) \left| \theta - \theta' \right|^{\lambda_\mu} \ .
\end{equation}
Since $\lim_n \theta_n = \theta_\star \un$ almost-surely, the above discussion
implies that for any fixed $x$, $\lim_n \Xi(1,\theta_n,x) = \Xi(1,\theta_\star
\un, x)$ almost-surely on $\cA_m$.

\paragraph{Moment conditions} \label{sec:Xi:DependanceX}
It is easily seen (using again Lemma~\ref{lem:TCL:C2}) that there exists a
constant $C_m$ such that for any $\theta \in \cK$, $|\Xi(1,\theta,x)| \leq C_m
(1+|x|^2)$. Therefore, Lemma~\ref{lem:TCL:agreementTau} implies that $\int
|\Xi(1,\theta_\star \un,x)| d\pi_{\star}(x) < \infty$. In addition, for any
$\theta \in \cK$, $\alpha$ in a neighborhood of $1$ and $a>1$,
\[
\int |\Xi(1,\theta,x)|^a \pi_{\alpha,\theta}(dx) \leq C_m \left(1 + \int |x|^{2a}
  \pi_{\alpha,\theta}(dx)\right) \ .
\]
Lemma~\ref{lem:TCL:agreementTau} implies that there exists $a>1$ such that
\[
\sup_n \un_{\theta_n \in \cK} \int |\Xi(1,\theta_n,x)|^a
\pi_{\alpha_{n+1},\theta_n}(dx) < \infty \ .
\]
\paragraph{Conclusion}
We can apply Lemma~\ref{lem:fmp12}; we have a.s., $\lim_n \left| \int \Xi(1, \theta_n, x) d\pi_n(x) - \int \Xi(1, \theta_\star
  \un,x) d\pi_\star(x) \right| \un_{\cA_m} = 0$.
  
  \vspace{0.1cm}
  
\subsubsection{Third term: $\Xi(1,\theta_n, \phi_n) - \int \Xi(1,\theta_n, x) d \pi_n(x)$} We prove that for any $m\geq 1$
\small

\vspace{-0.2cm}

\[
\lim_n \gamma_n \  \EE{ \left| \sum_{k=1}^n \left\{ \Xi(1,\theta_k, \phi_k) - \int
        \Xi(1,\theta_k, x) d \pi_k(x)\right\}\right| \un_{\cA_m} } = 0 \ 
.
\]
\normalsize

Set $\sum_{k=1}^n \left\{ \Xi(1,\theta_k, \phi_k) - \int \Xi(1,\theta_k,
  x) d \pi_k(x)\right\} = \sum_{i=1}^3 \mathcal{T}_n^{(i)}$  
 
  \vspace{-0.2cm}

  \small
  
\begin{align*}
  \text{with }\,\mathcal{T}_n^{(1)} & =  \sum_{k=1}^n \left\{ \Xi(1,\theta_k, \phi_k) -  \Xi(1, \theta_{k-1}, \phi_k)  \right\} \\
  \mathcal{T}_n^{(2)} & = \sum_{k=1}^n \left\{ \Xi(1, \theta_{k-1}, \phi_k) - \int
    \Xi(1,\theta_{k-1},
    x) d \pi_{k-1}(x)\right\}  \\
  \mathcal{T}_n^{(3)} & = \int \Xi(1,\theta_{0}, x) d \pi_{0}(x) - \int
  \Xi(1,\theta_{n}, x) d \pi_{n}(x) \ .
\end{align*}
\normalsize

\paragraph{Term $\mathcal{T}_n^{(1)}$}
By (\ref{eq:TCL:proof:Xi:RegulTheta}),  there exists a constant $C_m$ such that
for any $k \geq m+1$, on the set $\cA_m$,
$
|\Xi(1,\theta_k,\phi_k) - \Xi(1,\theta_{k-1},\phi_k) | \leq C_m |\theta_k -
\theta_{k-1}|^{\lambda_\mu} (1+ |\phi_k|^2) $.
Hence, by Lemma~\ref{lem:deltaTheta}, on the set $\cA_m$,
$  |\Xi(1,\theta_k,\phi_k) - \Xi(1,\theta_{k-1},\phi_k) | 
  \leq C_m \gamma_k^{\lambda_\mu} (1+ |\phi_k|^2) (|Y_k|^{\lambda_\mu} +
  |\phi_{k-1}|^{\lambda_\mu})$.
By Assumption~\ref{hyp:TCL:momentY}, Lemma~\ref{lem:TCL:agreementTau} and
Assumption~\ref{hyp:TCL:pas}, the sum
$
\sum_{k\geq 1} \gamma_k^{1+\lambda_\mu} \EE{ (1+ |\phi_k|^2)
  (|Y_k|^{\lambda_\mu} + |\phi_{k-1}|^{\lambda_\mu}) \un_{\cA_m} }
$
is finite which implies $\lim_n \gamma_n \EE{|\mathcal{T}_n^{(1)}|\un_{\cA_m}}
=0$ by the Kronecker Lemma.

\paragraph{Term $\mathcal{T}_n^{(2)}$}
From the expression of $\xi$ (see (\ref{eq:Expression:xi})), we have
$
\Xi(1,\theta,\phi) - \Xi(1,\theta,x) = \phi\T \C_{\theta} \phi - x\T \C_\theta
x + (\phi-x)\T \D_\theta
$
with $\C_\theta \eqdef  \int (w - \W_\theta) \A_{1,\theta}\T \A_{1,\theta} (w - \W_\theta) \,d \mu_\theta(y, w) $ and 
$  \D_\theta \eqdef 2 \int (w - \W_\theta) \A_{1,\theta}\T \A_{1,\theta} (wy -
  z_\theta)  \,d \mu_\theta(y, w)$.
We detail the proof of the statement

\vspace{-0.2cm}

\small

\[
\lim_n \gamma_n \EE{ \left| \sum_{k=1}^n \left( \phi_{k}- \int x \,
      d\pi_{\alpha_k,\theta_{k-1}}(x) \right)\T \D_{\theta_{k-1}} \right| \un_{\cA_m}} =0
\]

\normalsize

The second statement, with the quadratic dependence on $\phi_k$ is similar and
omitted (its proof will use
Proposition~\ref{prop:Regularity}-\ref{prop:Regularity:item3}) and the
condition $\lim_n \gamma_n n^{1/(1+\tau/2)} = 0$). Using again the Poisson
solution $f_n \eqdef f_{\alpha_{n+1},\theta_n}$ associated to the identity function and
the kernel $P_n \eqdef P_{\alpha_{n+1},\theta_n}$, it holds by
(\ref{eq:Proof:TCL:poisson})

\vspace{-0.2cm}

\small

\begin{align}
  & \left( \phi_{k}- \int x \, d\pi_{k-1}(x) \right)\T \D_{\theta_{k-1}}   \nonumber \\
  & \quad = \left( f_{k-1}(\phi_k) - P_{k-1} f_{k-1}(\phi_{k-1})
  \right)\T \, \D_{\theta_{k-1}} \label{proof:TCL:Tn2:bout1} \\
  & \quad + P_{k-1} f_{k-1}\T(\phi_{k-1}) \D_{\theta_{k-1}} - P_{k} f_{k}\T(\phi_{k}) \D_{\theta_{k}} \label{proof:TCL:Tn2:bout2}  \\
  & \quad + \left( P_{k} f_{k}\T(\phi_{k})  - P_{k-1} f_{k-1}\T(\phi_k) \right) \, 
  \D_{\theta_{k}} \label{proof:TCL:Tn2:bout3}  \\
  & \quad + P_{k-1} f_{k-1}\T(\phi_k) \left( \D_{\theta_{k}} - \D_{\theta_{k-1}}\right) \label{proof:TCL:Tn2:bout4}  \ .
\end{align}

\normalsize

From
Assumption~\ref{hyp:TCL:ContinuiteMuTheta}-\ref{hyp:TCL:ContinuiteMuTheta:item2bis})
and Lemma~\ref{lem:TCL:C2}, there exists a constant $C_m$ such that for any
$k$, 

\vspace{-0.2cm}

\small

\begin{align}
 & |\D_{\theta_k}| \un_{\cA_m} \leq C_m \label{proof:TCL:Tn2:bount5}  \\
& |\D_{\theta_k} - \D_{\theta_{k-1}}| \un_{\cA_m} \leq C_m \, |\theta_k - \theta_{k-1}|^{\lambda_\mu} \label{proof:TCL:Tn2:bount6}  \ .
\end{align}

\normalsize

Let us control the first term (\ref{proof:TCL:Tn2:bout1}). Upon noting that it
is a martingale-increment, the Burkholder inequality (see \emph{e.g.} \cite[Theorem
2.10]{hall:heyde:1980}) applied with $p\leftarrow 2+\tau$ and
Lemma~\ref{lem:TCL:agreementTau} imply

\vspace{-0.3cm}

\small

\[
\E{\left| \sum_{k=1}^n \left( f_{k-1}(\phi_k) - P_{k-1} f_{k-1}(\phi_{k-1})
    \right)\T \, \D_{\theta_{k-1}} \right| \un_{\cA_m}} = O\left( \sqrt{n}
\right)\ .
\]

\normalsize

This term is $o(1/\gamma_n)$ by Assumption~\ref{hyp:TCL:pas}. Let us consider
(\ref{proof:TCL:Tn2:bout2}).

\vspace{-0.3cm}

\small

\begin{align*}
 & \E{\left| \sum_{k=1}^n \left( P_{k-1} f_{k-1}\T(\phi_{k-1}) \D_{\theta_{k-1}}
        - P_{k} f_{k}\T(\phi_{k}) \D_{\theta_{k}} \right) \right| \un_{\cA_m}} \\
  & = \E{\left| P_{0} f_{0}\T(\phi_{0}) \D_{\theta_{0}} - P_{n}
      f_{n}\T(\phi_{n}) \D_{\theta_{n}} \right| \un_{\cA_m}}
\end{align*}

\normalsize

and this term is $O(1)$ by
Proposition~\ref{prop:invariantPI}-\ref{prop:invariantPI:item3}),
\eqref{proof:TCL:Tn2:bount5} and Lemma~\ref{lem:TCL:agreementTau}. Let us see
the third term~(\ref{proof:TCL:Tn2:bout3}). By
Proposition~\ref{prop:Regularity}-\ref{prop:Regularity:item2}) and (\ref{proof:TCL:Tn2:bount5}), we have

\vspace{-0.3cm}

\small

\begin{align*}
  & \E{\left| \sum_{k=1}^n \left( P_{k} f_{k}\T(\phi_{k}) - P_{k-1}
        f_{k-1}\T(\phi_k) \right) \, \D_{\theta_{k}} \right| \un_{\cA_m}} \\
  & \leq C_m \, \sum_{k=1}^n \E{\left(\left| \theta_k - \theta_{k-1}
    \right|^{\lambda_\mu}  + |\alpha_{k+1} - \alpha_k| \right) \un_{\cA_m}} 
\end{align*}

\normalsize

By Lemmas~\ref{lem:deltaTheta} and ~\ref{lem:TCL:agreementTau} and
Assumptions~\ref{hyp:TCL:momentY} and \ref{hyp:TCL:pas}, this term is
$o(1/\gamma_n)$. Finally, the same conclusion holds for
(\ref{proof:TCL:Tn2:bout4}) by using
Proposition~\ref{prop:invariantPI}-\ref{prop:invariantPI:item3}),
Lemma~\ref{lem:TCL:agreementTau} and (\ref{proof:TCL:Tn2:bount6}).  This
concludes the proof of $\lim_n \gamma_n \E{\left[| \mathcal{T}_n^{(2)}|
    \un_{\cA_m} \right]} =0$.

\paragraph{Term $\mathcal{T}_n^{(3)}$}
By Lemma~\ref{lem:TCL:C2}, there exists $C_m$ such that for any $\theta \in
\cK$, $|\Xi(1,\theta,x) | \leq C_m (1+|x|^2)$. By
Lemma~\ref{lem:TCL:agreementTau}, for any $a$ in a neighborhood of $1$ we have
$\sup_{\alpha \in [0,a], \theta \in \cK} \int |x|^2 \pi_{\alpha,\theta}(dx) <
\infty$. Since $\lim_n \alpha_n =1$, we have
$
\sup_{n \geq m} \left| \int \Xi(1,\theta_{n}, x) d \pi_{n}(x) \right|
\un_{\theta_n \in \cK} < C
$
for some constant $C$, which implies that $\lim_n \gamma_n
\EE{\left|\mathcal{T}_n^{(3)} \right| \un_{\cA_m}} =0$.

%%%%%%%%%%%%%%%%%%%%%%%%%%%%%%%%%%%%%%%%%%%%%%%%%%%%%%%%%%%%%%%%%%%%%%%%%%%%%%%
\vspace{-0.2cm}

\bibliographystyle{IEEEtran}
\bibliography{biblio}

\end{document}